\documentclass[review,onefignum,onetabnum]{siamart190516}
\usepackage{amssymb,mathrsfs,amsmath,paralist}
\usepackage{units}
\usepackage{xypic}
\usepackage{multicol}
\usepackage{xcolor}
\usepackage{graphicx}
\usepackage{amsmath}
\usepackage{cleveref}

\providecommand{\corollaryname}{Corollary}
\providecommand{\factname}{Fact}
\providecommand{\lemmaname}{Lemma}
\providecommand{\theoremname}{Theorem}
\providecommand{\claimname}{Claim}

\DeclareMathOperator{\pe}{\mathrm{p}}
\DeclareMathOperator{\de}{\mathrm{d}}
\newcommand{\BinC}{\operatorname{Bin-Clique}_k^n}

\newcommand{\NP}{\operatorname{NP}}
\newcommand{\ILP}{\operatorname{ILP}}
\newcommand{\GI}{\operatorname{GI}}
\newcommand{\LOP}{\operatorname{LOP}_n}

\newcommand{\bLOP}{\operatorname{Bin-LOP}_n}
\newcommand{\OP}{\operatorname{OP}_n}
\newcommand{\bOP}{\operatorname{Bin-OP}_n}
\newcommand{\UnC}{\operatorname{Clique}_k^n}
\newcommand{\Cl}{\operatorname{Clique}}
\newcommand{\RES}{\operatorname{Res}}
\newcommand{\PHP}[0]{\ensuremath{\textsc{PHP}}}

\newcommand{\LS}[0]{\ensuremath{\textsc{LS}}}

\newcommand{\val}[0]{\ensuremath{\textsf{val}}}

\DeclareMathOperator{\pPHP}{\mathrm{PHP}}

\DeclareMathOperator{\BinPHP}{\mathrm{Bin-PHP}}

\newcommand{\tuple}[1]{\mathbf{#1}}
\newcommand{\SA}[0]{\ensuremath{\textsc{SA}}}

\nolinenumbers

\title{Proof complexity and the binary encoding of combinatorial principles\footnotemark[1]}
\author{Stefan Dantchev  \footnotemark[2] \and Nicola Galesi  \footnotemark[3] \and Abdul Ghani \footnotemark[2] \and
 Barnaby Martin \footnotemark[2]
}

\begin{document}
\maketitle

\renewcommand{\thefootnote}{\fnsymbol{footnote}}

\footnotetext[1]{This paper is an expanded version of ``Resolution and the binary encoding of combinatorial principles'' from the 34th Computational Complexity Conference (CCC) 2019 and ``Sherali-Adams and the binary encoding of combinatorial principles'' from the 14th Latin American Theoretical Informatics Symposium (LATIN) 2020.}
\footnotetext[2]{Department of Computer Science, Durham University, U.K.}
\footnotetext[3]{Dipartimento di Informatica, Sapienza Universit\`a Roma.}



\newtheorem{thm}{\protect\theoremname}
\newtheorem{lem}[thm]{\protect\lemmaname}
\newtheorem{fact}[thm]{\protect\factname}
\newtheorem{cor}[thm]{\protect\corollaryname}
\newtheorem{claim}[thm]{\protect\claimname}

\makeatother

\begin{abstract}
We consider proof complexity in light of the unusual \emph{binary} encoding of certain combinatorial principles. We contrast this proof complexity with the normal unary encoding in several refutation systems, based on Resolution and Integer Linear Programming.

We firstly consider $\RES(s)$, which is an extension of Resolution working on $s$-DNFs. We prove an exponential lower bound of $n^{\nicefrac{\Omega(k)}{\de(s)}}$ for the size of refutations of the binary version of the  $k$-Clique Principle in $\RES(s)$, where \textcolor{black}{$s=o((\log \log n)^{\nicefrac{1}{3}})$} and $\de(s)$ is a doubly exponential function. Our result improves that of Lauria \mbox{et al.} who proved a similar lower bound for $\RES(1)$, \mbox{i.e.} Resolution. For the $k$-Clique and other principles we study, we show how lower bounds in Resolution for the unary version follow from lower bounds in $\RES(\log n)$ for the binary version, so we start a systematic study of the complexity of proofs in Resolution-based systems for families of contradictions given in the binary encoding.

We go on to consider the binary version of the (weak) Pigeonhole Principle $\BinPHP^m_n$. We prove that  for any \textcolor{black}{$\delta,\epsilon>0$, $\BinPHP^m_n$ requires refutations of size  $2^{n^{1-\delta}}$ in $\RES(s)$ for $s=O(\log^{\frac{1}{2}-\epsilon}n)$}. Our lower bound {\color{black}cannot be improved substantially with the same method} since for $m \geq 2^{\sqrt{n\log n}}$ we can prove there are {\color{black} $2^{O(\sqrt{n \log n})}$  size refutations}  of $\BinPHP^m_n$  in $\RES(\log n)$. This is a consequence of the like upper bound for the unary weak Pigeonhole Principle of Buss and Pitassi.

We consider the Sherali-Adams (\SA) refutation system where we prove lower bounds for both rank and size. For the unary encoding of the Pigeonhole Principle and the Ordering Principle, it is known that linear rank is required for refutations in \SA, although both admit refutations of polynomial size. We prove that the binary encoding of the (weak) Pigeonhole Principle $\BinPHP^m_n$ requires exponentially-sized (in $n$) \SA\ refutations, whereas the binary encoding of the Ordering Principle admits logarithmic rank, polynomially-sized \SA\ refutations. 

We continue by considering a natural refutation system we call ``\SA+Squares'', intermediate between \SA\ and Lasserre (Sum-of-Squares). This has been studied under the name static-$\mathrm{LS}^\infty_+$ by Grigoriev et al. In this system, the unary encoding of the Linear Ordering Principle $\LOP$ requires {\color{black} $O(n)$} rank while the unary encoding of the Pigeonhole Principle becomes constant rank. Since Potechin has shown that the rank of $\LOP$ in Lasserre is $O(\sqrt{n}\log n)$, we uncover an almost quadratic separation between \SA+Squares and Lasserre in terms of rank. Grigoriev \mbox{et al.} noted that the unary Pigeonhole Principle has rank $2$ in \SA+Squares and therefore polynomial size. Since we show the same applies to the binary $\BinPHP^{n+1}_n$, we deduce an exponential separation for size between \SA\ and \SA+Squares.

\end{abstract}

\section{Introduction} Various fundamental combinatorial principles used in proof complexity may be given in first-order logic as sentences $\varphi$ with no finite models. Riis discusses in \cite{SorenGap} how to generate from $\varphi$ a family of CNFs, the $n$th of which encoding the claim that $\varphi$ has a model of size $n$, which are hence contradictions. 

{\color{black}Following Riis, it is typical to encode the existence of the witnesses in longhand with a big disjunction of the form $v_{\tuple{a},1} \vee \ldots \vee v_{\tuple{a},n}$,\footnote{Here $\tuple{a}$ is a sequence of universal variables preceding the single existential variable the disjunction is witnessing.} that we designate the \emph{unary encoding}. 
As can be observed, in the unary encoding a solitary true literal tells us which is the witness. However one can think to encode the existence of 
such witnesses {\em succinctly}  by using a \emph{binary encoding}: each witness $j$ in the model of size $n$ can be captured by
$\log n$ variables  $\omega_{\tuple{a},h}^{j_h}$ capturing the parity $j_h$ of each bit $h$ of the the binary encoding $\mathrm{bin}(j)$ of $j$. 
The binary encoding of combinatorial statements is a natural extension to propositional formulas of the notion of the {\em bit-graph representation} of functions. 

 One of the main aims of proof complexity is to find hard combinatorial properties whose propositional translation might  lead to hard-to-prove formulas. The complexity of proving  formulas in proof systems  is measured  as a function of the size (or other measures like, for instance, the maximal width in CNFs) of the formula to be proved.   Hence  combinatorial principles encoded in binary are interesting to study in proof complexity: on the one hand they  preserve the combinatorial structure of the  principle encoded, on the other hand  they give a more succinct  propositional representation  of the formula to be studied that could make easier the task of obtaining strong lower bounds. Additionally, the binary encoding is symmetric in true and false. This intuition was leading many recent works proving hardness results for the complexity of proofs in several distinct proof systems and for different proof complexity measures.

In light of this, Thapen and Skelley considered in \cite{TS11} the binary encoding of a combinatorial principle on $k$-turn games $\GI_3$ and proved an exponential lower bound for refuting $\GI_3$ in \emph{Resolution}.  Several other examples followed and  more recently  the binary encoding of the {\em Pigeonhole principle} has been considered in several works. In the work \cite{DBLP:conf/focs/HrubesP17}, it was used to prove new size lower bounds for Cutting Planes, by a new technique. In the work \cite{DBLP:journals/jsyml/AtseriasMO15}, it was used to prove lower bounds for $\RES(s)$ refutations (which involved the relativised version of the weak pigeonhole principle). In the very recent work \cite{DBLP:journals/eccc/GoosNPRSR20}, it is used for the generalisation and simplification of the  $\NP$-completeness of automatising Resolution \cite{DBLP:journals/jacm/AtseriasM20}. Finally, in another recent work \cite{DBLP:journals/eccc/ItsyksonR20}, where it is called the \emph{bit Pigeonhole Principle}, it is used in a proof of lower bounds for $k$-party communication complexity. However, binary encodings are meaningful to apply to other statements as well and also to other proof complexity measures. The work \cite{DBLP:journals/combinatorica/LauriaPRT17} solves an important open problem on  the complexity of proofs in Resolution of a combinatorial principle expressing the presence of a $k$-clique in graphs, in the case of a binary encoding.   
 Several techniques to  prove {\em space} proof complexity lower bounds were applied successfully on the binary encoding of principles \cite{DBLP:journals/siamcomp/FilmusLNRT15,DBLP:journals/jacm/BonacinaG15,DBLP:journals/siamcomp/BonacinaGT16}.   
 
  In all these cases, considering the binary encoding led to significant lower bounds in an easier way than for the unary case. 
Of course  the idea of considering succinct encodings is not new and is not limited to proof complexity. Use of the binary encoding in bounded arithmetic seems to predate its use in proof complexity.  Furthermore, since the succinctness of the encoding of the formulas might affect the running time of routines having formulas as input, it is no surprise that binary encodings  have been studied  systematically  in the ``dual'' applied area of SAT-solving \cite{Commander,JustynaBook}, where it is usual to try different encodings of the $1$-from-$n$ constraint to speed-up the running time of SAT-solvers both on satisfiable and unsatisfiable formulas. In \cite{JustynaBook,Walsh00}, what we call the binary encoding is referred to as \emph{logarithmic}. 

 Merging the results in  \cite{DantchevGalesiMartin,DBLP:conf/latin/DantchevGM20}, the central thrust of this work is to start a systematic study contrasting the proof complexity between the unary and binary encodings of natural combinatorial principles. 
To compare the complexity of proving propositional binary and unary encodings  we will consider several refutation systems,  three distinct combinatorial principles (and their variants) and different complexity measures.
One of our  main contributions is a  lower bound similar to that obtained in \cite{DBLP:journals/combinatorica/LauriaPRT17}  for the binary  principle expressing the presence of $k$-cliques in graphs,  for an extension of the Resolution system which allows  bounded conjunctions, $\RES(s)$. In obtaining this lower bound 
we devise a new technique to prove size lower bounds in $\RES(s)$ which is suitable for binary encodings and which we also successfully apply to the the case of the Pigeonhole Principle. 
\newline }

\section{Overview of the results}

We consider three main combinatorial principles to contrast binary and unary proof complexity: (1) the $k$-Clique Formulas,  $\UnC(G)$;
 (2) the (weak) Pigeonhole Principle $\pPHP^m_n$;   and (3)  the (Linear) Ordering Principle, ($\mathrm{L}$)$\mathrm{OP}_n$.  
 
 The \emph{$k$-Clique Formulas} introduced in  \cite{Beyersdorff:2013:TOCL,DBLP:journals/toct/BeyersdorffGLR12,DBLP:conf/coco/BeameIS01} are formulas stating that a given graph $G$ does have a $k$-clique and are therefore unsatisfiable when $G$ does not contain a $k$-clique.  The  Pigeonhole Principle states that a total mapping $f:[m] \rightarrow [n]$ has necessarily a collision when $m>n$. Its propositional formulation in the negation, $\pPHP^m_n$ is well-studied in proof complexity (see among others: \cite{Haken,DBLP:journals/siamcomp/SegerlindBI04,DantchevR01,DBLP:journals/jacm/Raz04,DBLP:journals/tcs/Razborov03,10.1007/3-540-46011-X_8,Ben-sasson99shortproofs,BussP97,BeyersdorffGL10,DBLP:conf/focs/BeameP96,DBLP:journals/iandc/AtseriasBE02,DBLP:journals/tcs/Atserias03,DBLP:journals/jcss/MacielPW02}). 
The $\mathrm{(L)OP}_n$ formulas encode the negation of the (Linear) Ordering Principle which asserts that each finite (linearly) ordered set has a maximal element and was introduced and studied, among others, in the works \cite{DBLP:journals/acta/Krishnamurthy85,DBLP:journals/acta/Stalmarck96,DBLP:journals/cc/BonetG01}.

\textcolor{black}{Our work spans different proof systems. In fact, they are all actually refutation systems, though we often use the terms interchangeably.}

\subsection{Resolution and Res(s)}
{\color{black}
$\RES(s)$ is a refutational proof system extending Resolution to $s$-bounded DNFs, introduced by Kraj\'{\i}\v{c}ek in \cite{kraresk}.  
As a generalisation of Resolution, the complexity of proofs in $\RES(s)$ for the unary encoding was largely analysed in several works
\cite{DBLP:journals/iandc/AtseriasBE02,DantchevR03,EGM,DBLP:journals/siamcomp/SegerlindBI04, DBLP:journals/cc/Alekhnovich11,RazborovAnnals}.    

A principal motivation for the present work is to approach  size lower bounds of refutations in Resolution for families of contradictions in the  usual unary  encoding,  by looking at the complexity of proofs in $\RES(s)$ for the corresponding families of contradictions where witnesses are given in the binary encoding.  This method is justified by our observation, specified in Lemmas~\ref{lem:reslog} and \ref{lem:php-un-bin}, that for a family of contradictions encoding a principle which is expressible as a $\Pi_2$ first-order  formula having no finite models, short $\RES(\log n)$ refutations of  their  binary encoding can be obtained from short Resolution refutations for the unary encoding.  
In light of this observation we begin with the study of the binary version of the $k$-Clique Formula. Indeed a significant  size lower bound for the unary version of the $k$-Clique  Formulas in full Resolution is a long-standing open problem.  At present such lower bounds  are known only for restrictions of  Resolution: in the  treelike case \cite{Beyersdorff:2013:TOCL}, and, in a recent major breakthrough, for the case of read-once (or regular) Resolution  \cite{DBLP:journals/corr/abs-2012-09476}. 
}


\subsection{Sherali-Adams}

It is well-known that questions on the satisfiability of propositional CNF formulas may be reduced to questions on feasible solutions for certain Integer Linear Programs (ILPs). In light of this, several ILP-based proof (more accurately, refutation) systems have been suggested for propositional CNF formulas, based on proving that the relevant ILP has no solutions. 
Typically, this is accomplished by relaxing an ILP to a continuous Linear Program (LP), which itself may have (non-integral) solutions, and then reconstraining this LP iteratively until it has a solution iff the original ILP had a solution \textcolor{black}{(which happens at the point the LP has no solution)}. 
Among the most popular ILP-based refutation systems are Cutting Planes \cite{Gomory1960,Chvatal73} and several others proposed by Lov\'asz and Schrijver \cite{LovaszS1991}.

Another method for solving ILPs was proposed by Sherali and Adams \cite{SheraliA90}, and was introduced as a propositional refutation system in \cite{StefanGap}. Since then it has been considered as a refutation system in the further works \cite{TCS2009,NarrowMaximallyLong}. The Sherali-Adams system (\SA) is of significant interest as a static variant of the Lov\'asz-Schrijver system without semidefinite cuts (\LS). It is proved in \cite{laurent01comparison} that the \SA\ \emph{rank} of a polytope, roughly the number of iterations that  must be reconstrained until it is empty, is less than or equal to its \LS\ rank; hence we may claim that with respect to rank \SA\  is at least as strong as \LS\ (though it is unclear whether it is strictly stronger).

The binary encoding implicitly enforces an at-most-one constraint on the witness at the same time as it does the at-least-one. \textcolor{black}{That is, it specifies a unique witness.} Another way to enforce this is with unary functional constraints of the form $v_{\tuple{a},1} + \ldots + v_{\tuple{a},n} = 1$ (\mbox{cf.} the unary functional encoding of Section~\ref{subsec:theory}), where $\tuple{a}$ comes from a sequence of universal variables preceding the single existential variable the sum is witnessing. \textcolor{black}{This contrasts with the standard unary encoding which would be of the form $v_{\tuple{a},1} + \ldots + v_{\tuple{a},n} \geq 1$. We paraphrase our new} variant as being (the unary) \emph{encoding with equalities} or ``\SA-with-equalities'' and study this variant explicitly.  

 \subsection{\SA+Squares}
 
 We continue by considering a refutation system we call \SA+Squares which is between \SA\ and Lasserre (Sum-of-Squares) \cite{Lasserre2001} (see also \cite{laurent01comparison} for comparison between these systems). \SA+Squares appears as Static-LS$^\infty_+$ in \cite{russians}, where \SA\ is denoted Static-LS$^\infty$. In this system one can always assume the non-negativity of (the linearisation of) any squared polynomial. In contrast to our system \SA-with-equalities, we will see that the rank of the unary encoding of the Pigeonhole Principle is 2, while the rank of the Ordering Principle is linear. We prove this by showing a certain moment matrix in positive semidefinite.
 
\subsection{Three combinatorial principles}

We will now delve more deeply into known and new results for our three combinatorial principles. These are depicted in a visually agreeable fashion in Tables~\ref{tab:1} and \ref{tab:2}. \textcolor{black}{The principles themselves will be introduced in the appropriate section, though there is a table at the end of the appendix in which they can be conveniently found together in both the unary and binary encodings.} Let us adopt the following convention, which we will exemplify with the Pigeonhole Principle. \PHP\ refers to the principle (independently of the coding of the witnesses), $\PHP^m_n$ refers to the unary encoding and $\BinPHP^m_n$ refers to the binary encoding.

\begin{table}[ht]
\begin{center}
\begin{tabular}{|c|c|c|c|}
\hline
$\RES(s)$ & unary  & binary \\
\hline
 &  &  \textcolor{black}{not fpt}  \\
$\mathrm{(Bin\mbox{-})Clique^k_n}$ & open   & $n^{\nicefrac{\Omega(k)}{\de(s)}}$ \\
&  &  Corollary~\ref{cor:BinCreal} \\
\hline
 & subexponential upper & almost exponential lower  \\
$\mathrm{(Bin\mbox{-})PHP}^m_n$ &  $2^{O(\sqrt{n\log n})}$ & $2^{n^{1-\delta}}$ \\
 & \cite{BussP97} & Theorem~\ref{thm:wphp} \\
\hline
 & polynomial upper & polynomial upper  \\
$\mathrm{(Bin\mbox{-})OP}_n$ &  $O(n^3)$ & $\textcolor{black}{O(n^3)}$ \\
 & \cite{DBLP:journals/acta/Stalmarck96} & Lemma~\ref{lem:bi-op-res} \\
\hline

\end{tabular}
\end{center}
\vspace{0.3cm}
\begin{center}
\begin{tabular}{|c|c|c|c|}
\hline
$\SA$ size & unary  & binary \\
\hline
 & quadratic upper  & exponential tight \\
$\mathrm{(Bin\mbox{-})PHP}^{n+1}_n$  & $O(n^2)$ & $2^{\Theta(n)}$ \\
 & \cite{Rhodes07} & Corollary \ref{cor:size-bound} \\
\hline
\end{tabular}
\end{center}
\vspace{0.3cm}

\begin{center}
\begin{tabular}{|c|c|c|c|}
\hline
$\SA$ rank & unary  & binary \\
\hline
 & linear tight & logarithmic upper \\
$\mathrm{(Bin\mbox{-})LOP}_n$ & $n-2$ & $2\log n$ \\
& \cite{TCS2009} & Corollary~\ref{corollary:9} \\
\hline
\end{tabular}
\end{center}
\vspace{0.3cm}
\caption{Comparison of proof complexity between unary and binary encodings. In the first table, $\de(s)$ is a doubly exponential function and \textcolor{black}{consider $m$ to be} exponential in $n$. \textcolor{black}{A fixed parameter tractable (fpt) complexity takes the form of $f(k)n^{O(1)}$ and is ruled out by our result for $\mathrm{Bin\mbox{-}Clique^k_n}$ in $\RES(s)$.}}
\label{tab:1}
\end{table}

\begin{table}[H]
\begin{center}
\begin{tabular}{|c|c|c|c|c|}
\hline
unary rank & \SA & \SA-with-equalities & \SA+Squares & Lasserre \\
\hline
& linear & linear & constant & constant \\
$\PHP^{n+1}_n$ & tight & tight & &  \\
& \cite{TCS2009} &\cite{TCS2009} & \cite{russians} & \cite{russians} \\
\hline 
& linear & constant & linear & square root \\
$\mathrm{LOP}_n$ & tight &  & tight & almost tight \\
 & \cite{TCS2009} & Theorem \ref{theorem:8} & Theorem \ref{theorem:11} & \cite{potechin2020sum} \\
\hline
\end{tabular}
\end{center}
\vspace{0.3cm}
\begin{center}
\begin{tabular}{|c|c|c|c|c|}
\hline
binary size & \SA& \SA+Squares & Lasserre \\
\hline
& exponential & polynomial & polynomial\\
$\mathrm{Bin\mbox{-}}\PHP^{n+1}_n$ & lower & upper & upper \\
 & Theorem \ref{theorem:size-bound} & Theorem~\ref{thm:bin-php-sa-squares} & a fortiori\\
\hline 
 & polynomial  & polynomial & polynomial \\
$\mathrm{Bin\mbox{-}LOP}_n$ & upper & upper & upper \\
& Corollary \ref{corollary:9}  & a fortiori & a fortiori \\
\hline
\end{tabular}
\end{center}
\vspace{0.3cm}
\caption{A comparison of rank/degree and size for our principles in Sherali-Adams and its relatives. \textcolor{black}{Here by, e.g., `linear' we mean in the parameter $n$ parameterising both families, and not the number of variables.} }
\label{tab:2}
\end{table}

\subsubsection{The $k$-Clique Formulas}

{\color{black} Deciding whether a graph has a $k$-clique is an important  computational problem considered within computer science and its applications.} It can be decided in time 
$n^{O(k)}$ by a brute force algorithm. It is then of the utmost  importance to understand whether given algorithmic primitives are sufficient to design algorithms solving the Clique problem more efficiently than the trivial upper bound.  
Resolution refutations for the formula $\UnC(G)$ (respectively any CNF $F$), 
can be thought \textcolor{black}{of} as the execution trace of an algorithm, whose primitives are defined by the rules of the Resolution system, 
searching for a $k$-clique inside $G$ (respectively deciding the satisfiability of $F$).  Hence  understanding whether there are $n^{\Omega(k)}$ size lower bounds in Resolution for refuting $\UnC(G)$ would then 
answer the above question for algorithms based on Resolution primitives.  
\textcolor{black}{This question was posed in \cite{Beyersdorff:2013:TOCL} where they  proved that for canonical graphs  not containing 
$k$-cliques, that is $k-1$-partite complete graphs,  $\UnC(G)$ can be refuted efficiently, that is in size $O(n^22^k)$. 
In looking for classes of graphs making hard the formula $\UnC(G)$ for Resolution, \cite{Beyersdorff:2013:TOCL} considered the case when  $G$ is a random graph obtained by the Erd\H{o}s-R\'enyi distribution on graphs.  
For graphs $G$ in this family, they proved that $\UnC(G)$  requires $n^{\Omega(k)}$ size refutations in treelike Resolution, obtaining the desired lower bound but only for refutations restricted to tree form. Whether the lower bound for $\UnC(G)$ holds for general DAG-like  Resolution when $G$ is a  Erd\H{o}s-R\'enyi random graph is a major open problem which motivates this paper and towards which we contribute. 
This specific problem acquired even more importance as a consequence of two more recent results. On the one hand very recently Atserias et al. in \cite{DBLP:conf/stoc/AtseriasBRLNR18} proved an $n^{\Omega(k)}$ lower bound for $\UnC(G)$ when $G$ is a  Erd\H{o}s-R\'enyi random graph for the case of {\em read-once} Resolution refutations, that is a restriction of DAG-like Resolution, where each variable can be resolved at most once along any path in the refutation.  On the other hand in  the work \cite{DBLP:journals/combinatorica/LauriaPRT17}, Lauria et al. consider the binary encoding of Ramsey-type propositional statements, having as a special case a binary version of $\UnC(G)$: $\BinC(G)$. For this binary $k$-Clique Formula they obtain optimal $n^{\Omega(k)}$ size lower bounds for unrestricted Resolution. }

\textbf{Our Results}. We prove (in Corollary~\ref{cor:BinCreal}) an  $n^{\nicefrac{\Omega(k)}{\de(s)}}$ lower bound for the size of refutations of $\BinC$ in $\RES(o((\log \log n)^{\nicefrac{1}{3}}))$, where $\de(s)$ is a doubly exponential function and $G$ is a random  graph as defined in \cite{Beyersdorff:2013:TOCL}. 

\subsubsection{The (weak) Pigeonhole Principle} 

Lower bounds for $\RES(s)$ have appeared variously in the literature for the (weak) Pigeonhole Principle. Of most interest to us are those for the (moderately weak) Pigeonhole Principle $\pPHP^{2n}_n$, for $\RES(\sqrt{\log n / \log \log n})$ in \cite{DBLP:journals/siamcomp/SegerlindBI04}, improved to $\RES(\epsilon \log n / \log \log n)$ in \cite{RazborovAnnals}. Additionally, Buss and Pitassi, in \cite{BussP97}, proved  an upper bound of $2^{O(\sqrt{n\log n})}$ for the size of refuting $\pPHP^m_n$ in $\RES(1)$ when $m \geq 2^{\sqrt{n\log n}}$. 

{\color{black}In \cite{DBLP:journals/jsyml/AtseriasMO15}, an optimal lower bound is proven for the binary encoding of a relativised version of the pigeon-hole principle in $\RES(\log )$. Their technique, however, heavily depends on the relativisation and the specific choice of the parameters: no set of $\alpha n$ pigeons out of $n^\beta$ in total can be consistently mapped onto $n$ holes for any $\alpha, \beta > 1$. Proving a similar lower bound for the standard, unrelativised, version is a big question that remains wide open.
} 

In \cite{TCS2009} Dantchev et al. have proved that the \SA\ rank of (the polytope associated with) $\pPHP^{n+1}_n$ is $n-2$ (where $n$ is the number of \textcolor{black}{holes}). That there is a polynomially-sized refutation in \SA\ of $\pPHP^{n+1}_n$ is noted in  \cite{Rhodes07}. Grigoriev \mbox{et al.} have noted in \cite{russians} that there is a rank 2 and polynomially sized refutation of $\pPHP^{n+1}_n$ in Lasserre, and it is straightforward to see that this may be implemented in \SA+Squares.


\textbf{Our Results}. We prove that in $\RES(s)$, for all $\epsilon >0$ and $s\leq \log^{\textcolor{black}{\frac{1}{2}-\epsilon}}(n)$, the shortest proofs of $\BinPHP^m_n$, require size  $2^{n^{1-\delta}}$, for any $\delta>0$ (Theorem \ref{thm:wphp}). This is the first size lower bound known for the  $\BinPHP^m_n$ in $\RES(s)$. As a by-product of this lower bound we prove a lower bound of the order $2^{\Omega(\frac{n}{\log n})}$ (Theorem \ref{thm:reswphp}) for the size of the shortest Resolution refutation of $\BinPHP^m_n$.  Our lower bound for $\RES(s)$ is obtained through a technique that merges together the random restriction method, an inductive argument on the $s$ of $\RES(s)$ and the notion of {\em minimal covering} of a $k$-DNF of \cite{DBLP:journals/siamcomp/SegerlindBI04}.  

Since we are not using any (even weak) form of Switching Lemma (as for instance in  \cite{DBLP:journals/siamcomp/SegerlindBI04,DBLP:journals/cc/Alekhnovich11}), we consider how tight is our lower bound in  $\RES(s)$. 
We prove that $\BinPHP^m_n$  (Theorem \ref{thm:phpubtreelike}) can be refuted in size $2^{O(n)}$  in treelike $\RES(1)$.  This upper bound contrasts with the unary case, $\pPHP^m_n$, which instead requires  treelike $\RES(1)$ refutations of size $2^{\Omega(n \log n)}$, as proved in \cite{BeyersdorffGL10, DantchevR01}.

For the Pigeonhole Principle, similarly to the $k$-Clique Principle, we can prove that short $\RES(\log n)$ refutations for $\BinPHP^m_n$ can be efficiently obtained from short $\RES(1)$ refutations of $\pPHP^m_n$ (Lemma \ref{lem:php-un-bin}). This allows us to prove that our lower bound is almost optimal: from the aforementioned result of Buss and Pitassi \cite{BussP97} we deduce an exponential lower bound is not possible for $\BinPHP^m_n$ in $\RES(\log n)$.

We prove that the binary encoding $\BinPHP^m_n$ requires exponential size in \SA\ (Theorem~\ref{theorem:size-bound}), contrasting with the mentioned polynomially-sized refutations of the unary $\pPHP^m_n$. Finally, we prove that  $\BinPHP^m_n$ has polynomially sized and rank 2 refutations in \SA+Squares (Theorem~\ref{thm:bin-php-sa-squares}), in line with the corresponding result for the unary Pigeonhole Principle from \cite{russians}.
 


\subsubsection{Ordering Principles}

The Linear ordering formulas $\LOP$ claim that a linear ordering of some domain has no minimal element. In the case of finite domains, it is false.
They were used in \cite{DBLP:journals/cc/BonetG01,DBLP:journals/tocl/GalesiL10} as families of formulas witnessing  the optimality of the 
size-width tradeoffs for Resolution (\cite{Ben-sasson99shortproofs}), so that they require high width to be refuted, but still admit polynomial size refutations in Resolution. {\color{black} If we  drop the stipulation that the order is linear (total), we call the the principle $\OP$}.

In \cite{TCS2009} we showed that the \SA\ rank of (the polytope associated with) $\LOP$ is $n-2$. Since it is known that \SA\ polynomially simulates Resolution (see \mbox{e.g.} \cite{TCS2009}), it follows there is a polynomially-sized refutation in \SA\ of $\LOP$. Potechin has proved that $\LOP$ has refutations in Lasserre of degree $O(\sqrt{n}\log n)$. Though he uses a different version of $\LOP$ from us, we will see that his upper bound still applies.

\textbf{Our Results}. 
Firstly, we prove that $\bOP$ is polynomially provable in Resolution. 
Secondly, and in the world of \SA, we prove that the (unary) encoding of the Ordering Principle with equalities has rank 2 and polynomial size. This allows us to prove that $\bLOP$ has \SA\ rank at most $2 \log n$ and polynomial size. We prove a rank lower bound in \SA+Squares for $\LOP$ of $\Omega(n)$, thus giving a quadratic separation in terms of rank between \SA+Squares and Lasserre.

{\color{black}
\subsection{Main technical contributions}
As observed, one of the principal contributions of this work is the  $n^{\nicefrac{\Omega(k)}{\de(s)}}$ size lower  bounds for $\RES(s)$ refutations of  
$\BinC(G)$ when $G$ is a random graph as, for example, defined in \cite{Beyersdorff:2013:TOCL}. The interest of this lower bound lies in the fact that the Resolution complexity of $\UnC(G)$  at present is unknown and, as we prove in this paper, this  lower bound would follow from
a meaningful lower bound for $\BinC(G)$  in $\RES(\log n)$. Our result for $\RES(s)$ for  $\BinC(G)$ hence contributes towards this goal.     
 
The main mathematical tool  used so far to prove size lower bounds in $\RES(s)$ is a simplified version of the H\r{a}stad Switching Lemma \cite{hastad} which was introduced in the work of Buss, Segerlind and Impagliazzo \cite{DBLP:journals/siamcomp/SegerlindBI04} and later used  (and  slightly improved  in \cite{RazborovAnnals}) 
in all other works proving size lower bounds for $\RES(s)$ \cite{DBLP:journals/cc/Alekhnovich11}. Only for $\RES(2)$, in the work \cite{DBLP:journals/iandc/AtseriasBE02}, 
there is an example of a size lower bound using a random restriction method inherited from Resolution.
 
In this work we devise a recursive method to prove size lower bounds  in $\RES(s)$, which is especially suitable for binary principles and runs by recursion from  $s$ to $1$. Contrary to previous methods, our method does not use any form of the H\r{a}stad Switching Lemma. The main ingredients of our approach are:  (1)  special classes of random restrictions, which are especially suited for binary principles and can be easily composed recursively; (2)  the notion {of \em covering number} for a DNF (that is the minimal number of literals covering all the terms of a DNF), which was introduced in  \cite{DBLP:journals/siamcomp/SegerlindBI04}. 
The high level idea of the lower bound proof is as follows. Setting the covering number in the proper way, the recursion process applied on an allegedly small refutation a binary principle in $\RES(s)$ ends with a small $\RES(1)$, that is Resolution, refutation of a simplification of the same principle defined on  a smaller  but still  meaningful domain. At this point it is sufficient to prove (or to use if known) a size  lower bound for the principle in Resolution.

The lower bound for the $k$-Clique Formulas in  $\RES(s)$ is obtained by capturing a hardness property for the $k$-Clique Formulas which closely follows those defined  in \cite{Beyersdorff:2013:TOCL} for the unary case and later used and extended in \cite{MassimoRamsey,DBLP:conf/stoc/AtseriasBRLNR18}. However, differently from previous lower bounds, we isolate the hardness property in a definition
(see Definition \ref{def:extension}) and a lemma called the {\em Extension Lemma} (see Lemma \ref{lem:extension}), whose aim is that of capturing the existence of non-trivial families of partial assignments that applied to the $k$-Clique Formula do not trivialise its Resolution refutations.  This is inspired  by the Atserias-Dalmau \cite{DBLP:journals/jcss/AtseriasD08} approach to prove width lower bounds (and hence size lower bounds) for Resolution. 
}

\subsection{Contrasting  unary and binary principles} 
\label{subsec:theory}

We go on to consider the relative properties of unary and binary encodings, especially for Resolution. We take the case in which the principle is binary and involves total comparison on all its relations. That is, where there are axioms of the form $v_{i,j} \oplus v_{j,i}$, where $\oplus$ indicates XOR, for each $i \neq j$. We argue that the proof complexity in Resolution of such principles will not increase significantly (by more than a polynomial factor) when shifting from the unary encoding to the binary encoding.

The \emph{unary functional} encoding of a combinatorial principle replaces the big disjunctive clauses of the form $v_{i,1} \vee \ldots \vee v_{i,n}$, with $v_{i,1} + \ldots + v_{i,n} = 1$, where addition is made on the natural numbers. We already met this in the context of \SA, but it is equivalent to augmenting the axioms $\neg v_{i,j} \vee \neg v_{i,k}$, for $j \neq k \in [n]$. One might argue that the unary functional encoding is the true unary analog to the binary encoding, since the binary encoding naturally enforces that there is a single witness alone. It is likely that the non-functional formulation was preferred for its simplicity (similarly as the Pigeonhole Principle is often given in its non-functional formulation).

In Subsection \ref{subsec:func}, we prove that the Resolution refutation size increases by only a quadratic factor when moving from the binary encoding to the unary functional encoding. This is interesting because the same does not happen for treelike Resolution, where the unary encoding of the Pigeonhole Principle has complexity $2^{\Theta(n \log n)}$  \cite{BeyersdorffGL10, DantchevR01}, while, as we prove in Subsection \ref{subsec:treephp} (Theorem \ref{thm:phpubtreelike}), the binary (functional) encoding is $2^{\Theta(n)}$. The unary encoding complexity is noted in \cite{DantchevR03} and remains true for the unary functional encoding with the same lower-bound proof. The binary encoding complexity is addressed directly in this paper.  


\subsection{Structure of the paper}

After the preliminaries in Section~\ref{sec:pre}, we move on to the $\RES(s)$ lower bounds for $\BinC$ in Section~\ref{sec:k-clique} and $\BinPHP^m_n$ in Section~\ref{sec:wphp}. In Section~\ref{sec:SAbinPHP} we prove our \SA\ size lower bound for $\BinPHP^m_n$ and in Section~\ref{sec:SAeqLOP} we prove our \SA\ size and rank upper bounds for the Linear Ordering Principle with equalities, which apply, as a corollary, also to  to $\bLOP$. In Section~\ref{sec:SA+AS}, we introduce \SA+Squares and discuss upper bounds for $\PHP$ and give a lower bound for $\LOP$. In Section~
\ref{sec:unbin}, we make further comments on the constrast between unary and binary encodings in general for Resolution. In Section~\ref{sec:fin-rem}, we make some final remarks.

\textcolor{black}{Two objects inhabit an appendix.} 
\textcolor{black}{Firstly,} an argument that Potechin's Lasserre upper bound for $\LOP$ from \cite{potechin2020sum} applies also to our encoding. \textcolor{black}{Secondly, a table recapping the unary and binary encodings of the main principles}.

\section{Preliminaries}
\label{sec:pre}

Let $[n]$ be the set $\{1,\ldots,n\}$.
Let us assume, without loss of much generality, that $n$ is a power of $2$. Cases where $n$ is not a power of $2$ are handled in the binary encoding by explicitly forbidding possibilities. \textcolor{black}{Let $\mathrm{bin}(a)$ be the sequence $a_1\ldots a_{\log n}$ which is $a$ written in binary, say from the most significant digit to the least.}

If $v$ is a propositional variable, then $v^0=\neg v$ indicates the negation of $v$, while $v^1$ indicates $v$. 
We denote by $\top$ and $\bot$ the Boolean values {}``true'' and
{}``false'', respectively. A \emph{literal} is either a propositional
variable or a negated variable. We will denote literals by small
letters, usually $l$'s. An $s$\emph{-conjunction} ($s$-\emph 
{disjunction})
is a conjunction (disjunction) of at most $s$ literals. A {\em clause} with $s$ literals is a $s$-disjunction. The width
$w(C)$ of a clause $C$ is the number of literals in $C$. A \emph{term}
($s$\emph{-term}) is either a conjunction ($s$-conjunction) or a
constant, $\top$ or $\bot$. An $s$-\emph{DNF} or $s$\emph{-clause}
($s$\emph{-CNF}) is a disjunction (conjunction) of an unbounded number
of $s$-conjunctions ($s$-disjunctions). We will use calligraphic
capital letters to denote $s$-CNFs or $s$-DNFs, usually ${\cal C}$s for CNFs, ${\cal D}$s for DNFs and ${\cal F}$s for both.
For example, $((v_1 \wedge \neg v_2)\vee (v_2 \wedge v_3) \vee (\neg v_1 \wedge v_3))$ is an example of a $2$-DNF and its negation $((\neg v_1 \vee  v_2)\wedge (\neg v_2 \vee \neg v_3) \wedge (v_1 \vee \neg v_3))$ is an example of a $2$-CNF.

\subsection{Res(s) and Resolution}

We can now describe the propositional refutation system $\RES\left(s\right)$ (\cite{krabook95}). It is used to \emph{
refute} (\mbox{i.e.} to prove inconsistency) of a given set of $s$-clauses
by deriving the empty clause from the initial clauses. There are four
derivation rules:

\begin{enumerate}
\item The $\wedge$-\emph{introduction rule} is\[
\frac{\mathcal{D}_{1}\vee\bigwedge_{j\in J_{1}}l_{j}\quad\mathcal{D}_ 
{2}\vee\bigwedge_{j\in J_{2}}l_{j}}{\mathcal{D}_{1}\vee\mathcal{D}_{2} 
\vee\bigwedge_{j\in J_{1}\cup J_{2}}l_{j}},\]
provided that $\left|J_{1}\cup J_{2}\right|\leq s$.
\item The \emph{cut (or resolution) rule} is\[
\frac{\mathcal{D}_{1}\vee\bigvee_{j\in J}l_{j}\quad\mathcal{D}_{2}\vee 
\bigwedge_{j\in J}\neg l_{j}}{\mathcal{D}_{1}\vee\mathcal{D}_{2}}.\]

\item The two \emph{weakening rules} are\[
\frac{\mathcal{D}}{\mathcal{D}\vee\bigwedge_{j\in J}l_{j}}\quad\textrm 
{and}\quad\frac{\mathcal{D}\vee\bigwedge_{j\in J_{1}\cup J_{2}}l_{j}} 
{\mathcal{D}\vee\bigwedge_{j\in J_{1}}l_{j}},\]
provided that $\left|J\right|\leq s$.
\end{enumerate}
A $\RES(s)$ refutation can be considered as a directed acyclic
graph (DAG), whose sources are the initial clauses, called also axioms,
and whose only sink is the empty clause. We shall define \emph{the
size of a proof} to be the number of internal nodes of the graph,
\mbox{i.e.} the number of applications of a derivation rule, thus ignoring
the size of the individual $s$-clauses in the refutation. In principle the $s$ from {}``$\RES(s)$'' could depend
on $n$ --- an important special case is $\RES(\log n)$.

Clearly, $\RES(1)$ is \emph{(ordinary) Resolution}, working
on clauses, and using only the cut rule, which becomes the usual
resolution rule, and the first weakening rule. Given an unsatisfiable CNF ${\cal C}$, 
and a $\RES(1)$ refutation $\pi$ of ${\cal C}$ the width of $\pi$, $w(\pi)$ is the  maximal width of a  clause in $\pi$.
The width  of refuting ${\cal C}$ in Res(1), $w(\vdash {\cal C})$, is the minimal width over all $\RES(1)$ refutations of ${\cal C}$.

A {\em covering set} for an $s$-DNF ${\cal D}$  is a set of literals $L$ such that each term of ${\cal D}$ has at least one literal in $L$.
The {\em covering number} $c({\cal D})$ of an $s$-DNF ${\cal D}$ is the minimal size of a covering set for ${\cal D}$.  \textcolor{black}{We extend
the definition of covering number to the case of $s$-CNFs: the covering number of a $s$-CNF $F$ is the covering 
number of the DNF obtained by applying De Morgan simplifications to $\neg F$.}   

Let ${\cal F}(v_1\ldots,v_n)$ be a boolean $s$-DNF (resp. $s$-CNF) defined over 
variables $V=\{v_1,\ldots,v_n\}$. A  {\em partial assignment} $\rho$ to ${\cal F}$  is a truth-value assignment to some 
of the variables of ${\cal F}$: $dom(\rho) \subseteq V$.  By ${\cal F}\!\!\!\upharpoonright_\rho$ we denote the formula ${\cal F}'$ over variables in $V\setminus dom(\rho)$ 
obtained from ${\cal F}$ after simplifying in it the variables in $\mathit{dom}(\rho)$ according to the usual boolean simplification rules of clauses and terms.

Similarly to what was done for treelike $\RES(s)$ refutations in \cite{EGM}, if we turn a $\RES(s)$ refutation of a given set
of $s$-clauses $\mathcal{F}$ upside-down, \mbox{i.e.} reverse the edges of the
underlying graph and negate the $s$-clauses on the vertices, we get
a special kind of restricted branching $s$-program whose nodes are labelled by $s$-CNFs and at each node some $s$-disjunction is questioned. The \textcolor{black}{restrictions placed on the branching program} are
as follows.

Each vertex is labelled by an $s$-CNF which partially represents the  
information
that can be obtained along any path from the source to the vertex (this is a \emph{record} in the parlance of \cite{proofs_as_games}).
Obviously, the (only) source is labelled with the constant $\top$.
There are two kinds of queries that can be made by a vertex:

\begin{enumerate}
\item Querying a new $s$-disjunction, and branching on the answer, which
can be depicted as follows.\begin{equation}
\begin{array}{ccccc}
  &  & \mathcal{C}\\
  &  & ?\bigvee_{j\in J}l_{j}\\
  & \top\swarrow &  & \searrow\bot\\
\mathcal{C}\wedge\bigvee_{j\in J}l_{j} &  &  &  & \mathcal{C}\wedge 
\bigwedge_{j\in J}\neg l_{j}\end{array}\label{eq:new-query}\end 
{equation}

\item Querying a known $s$-disjunction, and splitting it according to  
the answer:
\begin{equation}
\begin{array}{ccccc}
  &  & \mathcal{C}{\wedge\bigvee}_{j\in J_{1}\cup J_{2}}l_{j}\\
  &  & ?\bigvee_{j\in J_{1}}l_{j}\\
  & \top\swarrow &  & \searrow\bot\\
\mathcal{C}\wedge\bigvee_{j\in J_{1}}l_{j} &  &  &  & \mathcal{C} 
\wedge\bigvee_{j\in J_{2}}l_{j}\end{array}
\label{eq:split-query}
\end{equation}

\end{enumerate}

\noindent There are two ways of forgetting information,
{\begin{equation}
\begin{array}{c}
{\cal C}_{1}\wedge{\cal C}_{2}\\
\downarrow\\
{\cal C}_{1}\end{array}\qquad\textrm{and}\qquad\begin{array}{c}
\mathcal{C}\wedge\bigvee_{j\in J_{1}}l_{j}\\
\downarrow\\
\mathcal{C}\wedge\bigvee_{j\in J_{1}\cup J_{2}}l_{j}\end{array},\label 
{eq:forget}\end{equation}
the point being that forgetting allows us to equate the information
obtained along two different branches and thus to merge them into
a single new vertex.} \textcolor{black}{For simplicity when calculating the size of refutation subtrees, let us assume that a weakening may be integrated into either side of a query.} A sink of the branching $s$-program must be labelled
with the negation of an $s$-clause from $\mathcal{F}$. Thus the branching $s$-program
is supposed by default to solve the \emph{Search Problem for $\mathcal{F}$}:
given an assignment of the variables, find a clause which is falsified
under this assignment.

The equivalence between a $\RES(s)$
refutation of $\mathcal{F}$ and a branching $s$-program of the kind above is
obvious. Naturally, if we allow querying single variables only, we
get branching $1$-programs -- decision DAGs -- that correspond to Resolution. If we do not
allow the forgetting of information, we will not be able to merge distinct
branches, so what we get is a class of decision trees that correspond
precisely to the treelike version of these refutation systems. The queries of the form (\ref{eq:new-query})
and (\ref{eq:split-query}) as well as forget-rules of the form (\ref 
{eq:forget})
give rise to a Prover-Adversary game (see \cite{proofs_as_games}
where this game was introduced for Resolution). In short, Adversary
claims that $\mathcal{F}$ is satisfiable, and Prover tries to expose him.
Prover always wins if her strategy is kept as a branching program
of the form we have just explained, whilst a good (randomised)  
Adversary's
strategy would show a lower bound on the branching program, and thus
on any $\mbox{Res}\left(s\right)$ refutation of $\mathcal{F}$.

\begin{lemma}
If a CNF $\phi$ has a refutation in $\RES(k+1)$ of size $N$, whose corresponding branching $(k+1)$-program has no $(k+1)$-CNFs of covering number $\geq d$, then $\phi$ has a $\RES(k)$ refutation of size $2^{\textcolor{black}{d+1}}\cdot N$ (which is $\leq e^d \cdot N$ when $d>4$).
\label{lem:covering-number}
\end{lemma}
\
\begin{proof}
\textcolor{black}{
In the branching program, consider a $(k+1)$-CNF $\phi$ whose covering number $< d$ is witnessed by variable set $V':=\{v_1,\ldots,v_{d-1}\}$. At this node some $(k+1)$-disjunction $(l_1 \vee \ldots \vee l_{k} \vee l_{k+1})$ is questioned.
}

\textcolor{black}{
Now in place of the CNF record $\phi$ in our original branching program we expand a mini-tree of size $2^{d+1}$ with $2^{d}$ leaves questioning all the variables of $V'$ as well as the literal $l_{k+1}$. Clearly, each evaluation of these reduces $\phi$ to a $k$-CNF that logically implies $\phi$. This may involve a weakening step in the corresponding $\RES(k)$ refutation. It remains to explain how to link the leaves of these mini-trees to the roots of other mini-trees. At each leaf we look to see whether we have the information $l_{k+1}$ or $\neg l_{k+1}$. If $l_{k+1}$ then we link immediately to the root of the mini-tree corresponding to the yes-answer to  $(l_1 \vee \ldots \vee l_{k} \vee l_{k+1})$ (without asking a question). If $\neg l_{k+1}$ then we question $(l_1 \vee \ldots \vee l_{k})$ and, if this is answered yes, link the yes-answer to  $(l_1 \vee \ldots \vee l_{k} \vee l_{k+1})$, otherwise to its no-answer.
}
\end{proof}


\subsection{Sherali-Adams via (integer) linear programming}
\color{black}

Following \cite{TCS2009} we define the $\SA$ proof system in a $\ILP$ form and  hence in terms of linear inequalities and we explain later the equivalence with an alternative definition by polynomials. 

Let $\mathcal{C}$ be a CNF $ C_1 \wedge \ldots \wedge C_m$ in variables $V=\{v_1,\ldots,v_n\}$.
Let $L_V =\{ v_1,\ldots,v_n, \neg v_1,\ldots,\neg v_n\}$ and adopt the convention that for $l \in L_V$, if 
$l = \neg v$ then $\bar l = v$ and if $l=v$, then $\bar l = \neg v$. 
First we  introduce a set of integer variables of the form $Z_D$,  where $D$ is a conjunction of {\em distinct} literals in $L_V$, 
with the meaning that  $Z_{\bigwedge_{i} \! l_i}$ is false if its subscript is false.\footnote{We are considering here $n$ new formal variables  $\bar V = \{\bar v_1,\ldots, \bar v_n\}$ such that $v =(1-\bar v)$. This allow us to compactly write a polynomial of the form $\prod_i(1-v_i)$ as a monomial $\prod_i \bar v_i$, modulo the set of polynomials stating that $v+\bar v=1$ taken for all variables $v$.  
}

We consider $Z_1=Z_\emptyset$, where $\emptyset$ is an empty conjunction, to be associated with the monomial equation $0=0$ and we assume that the {\em names} of the $Z$ variables 
fulfil the basic properties of the $\wedge$ operator such as commutativity and idempotence.  So, for instance,
 $Z_{D_1\wedge D_2}$ is the same variable as $Z_{D_2\wedge D_1}$, or $Z_{1 \wedge D}$ as well as $Z_{D \wedge D}$ are  both the variable $Z_D$.

\smallskip

For $0\leq r<2n$ let ${\cal D}_r$ be the set of the conjunctions of at most $r$ literals in $L_V$ (being $1$ the empty conjunction). 
We let  $\mathcal{P}_r^{\mathcal{C}}$ to be  the polytope specified by the following inequalities.

\begin{eqnarray}
0 \leq Z_{l \wedge D} \leq Z_D   \qquad  \qquad   \qquad     l  \in L_V, D \in {\cal D}_r \label{eq:LinSA1}\\
Z_{l \wedge D} +Z_{\bar l \wedge D} =Z_D    \qquad   \qquad  \qquad  l \in L_V, D \in {\cal D}_r \label{eq:LinSA2}\\
(Z_{\!D \wedge l_{1}} + \cdots + Z_{\!D \wedge l_{k}}) \geq Z_D \qquad (l_1\vee \ldots\vee  l_k) \in C,D \in {\cal D}_r \label{eq:LinSA3}
\end{eqnarray}

Observe that $\mathcal{P}_0^{\mathcal{C}}$, the polytope associated to $\mathcal{C}$, is specified by the inequalities 
$$
\left\{
\begin{array}{ll}
0 \leq Z_{l} \leq 1 &  \qquad l \in L_V \\
Z_{l} +Z_{\bar l} =1  & \qquad  l \in L_V\\
Z_{l_{1}} + \cdots + Z_{l_{k}} \geq 1 
& \qquad (l_1\vee \ldots \vee l_k) \in C
\end{array}
\right.
$$
\noindent It is clear that  ${\cal P}^{{\cal C}}_0$ contains integral $\{0,1\}$ points if and only if  $\mathcal{C}$ is satisfiable.   

\smallskip

Sherali-Adams (\SA) is  a static refutation method that takes the polytope $\mathcal{P}^{\mathcal{C}}_0$ whose dimension
is  $2n$ 
and \emph{$r$-lifts} it, by the definition of new variables and constraints, to another polytope $\mathcal{P}^{\mathcal{C}}_r$ whose dimension  is $\sum_{\lambda=0}^{r+1} {2n\choose \lambda}$. 
Observe that on unsatisfiable CNFs $\mathcal{C}$, $\mathcal{P}^{\mathcal{C}}_0$ does not contain integral points but it is not necessarily empty,  while necessarily $\mathcal{P}^{\mathcal{C}}_{2n}$ is the empty polytope (indeed, already $\mathcal{P}^{\mathcal{C}}_{n-1}$ is empty). Hence the following definition is meaningful.

\begin{definition}
The $\SA$-{\em rank} of an unsatisfiable  CNF ${\cal C}$ (we equivalently say the $\SA$-{\em rank} of $\mathcal{P}^{\mathcal{C}}_0$) 
is the minimal $r\leq 2n$ such that $\mathcal{P}^{\mathcal{C}}_r$
is the empty polytope.  A $\SA$-{\em refutation} of ${\cal C}$ is a subset of constraints in the definition of 
${\cal P}^{\cal C}_r$ that defines an empty polytope. 
\end{definition}
\noindent Note that \SA\ is polynomially verifiable due to the tractability of linear programming.

Let us point out some simple properties we use later.  It is easy to see that for $r' \leq r$, the defining inequalities of $\mathcal{P}^{\mathcal{C}}_{r'}$ are included in   those of $\mathcal{P}^{\mathcal{C}}_{r}$. Hence any solution to the inequalities of $\mathcal{P}^{\mathcal{C}}_{r}$ gives rise to solutions of the inequalities of $\mathcal{P}^{\mathcal{C}}_{r'}$, when projected onto its variables. If $D'$ is a conjunction of $r'$ literals, then $Z_{D \wedge D'} \leq Z_{D}$ follows by transitivity from $r'$ instances of $(\ref{eq:LinSA1})$. We refer to the property $Z_{D \wedge D'} \leq Z_{D}$ as \emph{monotonicity}. Finally, let us note that $Z_{v \wedge \neg v}=0$ holds in $\mathcal{P}^{\mathcal{C}}_1$ and follows from a single lift of an equality of negation.

\textcolor{black}{Our use of distinct literals $Z_v$ and $Z_{\neg v}$, with the axioms (\ref{sec:pre}.2), is not followed in all expositions of Sherali-Adams as a refutation system \SA. Indeed, in \cite{AtseriasLN14}, the use of these so-called \emph{twin variables} begets a new refutation system labelled \textsc{SAR} (in an apparent homage to the \textsc{PCR} of \cite{AlekhnovichBRW02}). Note that the rank measure is equivalent in both versions of \SA, and size lower bounds, for our version with twin variables, are at least as strong as with the alternative version.}

	\subsubsection{Sherali-Adams via polynomials}
	
	Here we give an alternative definition of Sherali-Adams and explain its relation to the one just given.
	
	\begin{definition}
		A \emph{Sherali-Adams} refutation of a set of linear inequalities $a_1 \geq 1 , \ldots, a_m \geq 1$ over a set of variables $V$ is a formal equality of the form
		
		\begin{equation} \label{eq:SApoly}
			c_0 + \sum_{i = 1}^m c_i a_i = -1
		\end{equation}
		
\noindent		where each $c_i$ is a polynomial over $V$ with non-negative coefficients, and the multiplication is carried out over the quotient ring $\mathbb{R}^V / \{v^2 - v : v \in V\}$ (that is, idempotently). The \emph{degree} of the refutation is the maximum degree of the polynomials $c_i a_i$. The \emph{size} of the refutation is the total number of monomials appearing with nonzero coefficient on the left hand side of \cref{eq:SApoly}

	\end{definition}

	It is clear that Sherali-Adams is sound, in the sense that if a set of linear inequalities admits a Sherali-Adams refutation then it has no 0/1 solutions. Once the degree is fixed, the search for the coefficients of the $c_i$ in \Cref{eq:SApoly} can be formulated as a linear program. It can be seen that the dual of this program is exactly the definition given first (see, e.g., \cite{laurent01comparison}). Imagine, for some CNF $C$ over the variables $V = \{v_1, \ldots, v_n\}$ and some rank $r$, that $\mathcal{P}_r^{\mathcal{C}}$ is nonempty. Then pick some $x \in \mathcal{P}_r^{\mathcal{C}}$ and define a linear operator $\lambda$ on monomials of degree at most $r + 1$ defined by $\lambda(v_{x_1} \cdot v_{x_2} \cdots v_{x_d}) = x(Z_{v_{x_1} \wedge \ldots \wedge v_{x_d}})$. Then the set of inequalities gotten from sending each clause $l_1 \vee \ldots \vee l_k$ in $\mathcal{C}$ to $\sum_{i = 1}^k l_i \geq 1$ has no Sherali-Adams refutation of degree at most $r$, because then  $\lambda$ when applied to both sides of \cref{eq:SApoly} would produce a contradiction.

\section{Res(s) and the binary encoding of $k$-Clique}
\label{sec:k-clique}
\color{black}

Consider a graph $G$ such that $G$ is formed from $k$ blocks of $n$ nodes each: $G =(\bigcup_{b\in[k]} V_b, E)$, where edges may only appear between distinct blocks. Thus, $G$ is a $k$-partite graph. Let the 
edges in $E$ be denoted as pairs of the form $E((i,a),(j,b))$, where $i\neq j \in[k]$ and $a,b\in [n]$.

The (unary) $k$-Clique CNF formulas $\UnC(G)$ has variables $v_{i,q}$ with $i \in [k], q \in [n]$, with clauses $\neg v_{i,a} \vee \neg v_{j,b}$ whenever $\neg E((i,a),(j,b))$ (\mbox{i.e.} there is no edge between node $a$ in block $i$ and node $b$ in block $j$), and clauses $\bigvee_{a \in [n]} v_{i,a}$, for each block $i$. This expresses that $G$ has a $k$-clique (with one vertex in each block), which we take to be a contradiction, since we will arrange for $G$ not to have a $k$-clique.  \textcolor{black}{ Notice that this formula 
encodes the fact that the graph contains a {\em transversal
$k$-clique}, that is, a $k$-clique in which each node belongs to a different block. 
 As noticed in \cite{Beyersdorff:2013:TOCL, DBLP:conf/stoc/AtseriasBRLNR18} a graph can contain a $k$-clique but  no transversal $k$-clique for a given partition.
Finding a transversal $k$-clique in a given graph is intuitively more difficult then finding a $k$-clique, hence 
proving that a graph does not contain a transversal $k$-clique should be easier
than proving it does not contain any $k$-clique. This was formally proved to hold even for treelike Resolution 
(see Lemma 2.2 in \cite{DBLP:conf/stoc/AtseriasBRLNR18}). }

\medskip
 $\BinC(G)$  variables $\omega_{i,j}$ range over $i \in [k], j \in [\log n]$. Let us assume for simplicity of our exposition that $n$ is a power of $2$, the general case \textcolor{black}{requires the explicit forbidding of certain combinations.}
Let $a \in [n]$ and let $a_1 \ldots  a_{\log n}$ be \textcolor{black}{$\mathrm{bin}(a)$}.  
Each (unary) variable $v_{i,a}$ semantically corresponds to the  conjunction $(\omega^{a_1}_{i,1}\wedge \ldots \wedge \omega^{a_{\log n}}_{i,\log n})$,
where
$$
\omega^{a_j}_{i,j}=
\left\{
\begin{array}{ll}
\omega_{i,j} & \mbox{ if $a_j=1$} \\
\neg \omega_{i,j} & \mbox{ if $a_j=0$}
\end{array}
\right.
$$
Hence in $\BinC(G)$ we encode the unary clauses $\neg v_{i,a} \vee \neg v_{j,b}$,
by the clauses
$$
(\omega^{1-a_1}_{i,1}\vee \ldots \vee \omega^{1-a_{\log n}}_{i,\log n}) \vee (\omega^{1-b_1}_{j,1}\vee \ldots \vee \omega^{1-b_{\log n}}_{j,{\log n}}).
$$
{\color{black} Notice that the wide clauses $\bigvee_{a \in [n]} v_{i,a}$ from the unary encoding automatically become true under the binary encoding.}



By the next lemma short Resolution refutations for  $\UnC(G)$ can be translated into short $\RES(\log n)$ refutations of $\BinC(G)$. 
Hence to obtain lower bounds for $\UnC(G)$ in Resolution, it suffices to obtain lower bounds for 
 $\BinC(G)$ in $\RES(\log n)$. 

\begin{lemma}
\label{lem:reslog}
Suppose there are Resolution refutations of $\UnC(G)$ of  size $S$. 
Then there are $\RES(\log n)$ refutations of $\BinC(G)$  of size $S$.
\end{lemma}
\begin{proof}
Where the decision DAG for $\UnC(G)$ questions some variable $v_{i,a}$, the decision branching $\log n$-program questions instead 
$(\omega^{1-a_1}_{1,1}\vee \ldots \vee \omega^{1-a_{\log n}}_{1,{\log n}})$ where the out-edge marked true in the former becomes false in the latter, and vice versa. What results is indeed a decision branching $\log n$-program for $\BinC(G)$, and the result follows.
\end{proof}

Following  \cite{Beyersdorff:2013:TOCL,DBLP:conf/stoc/AtseriasBRLNR18,DBLP:journals/combinatorica/LauriaPRT17} we consider $\BinC(G)$ formulas  where $G$ is   a random graph distributed according to a variation of the Erd\H{o}s-R\'enyi distribution as defined in  \cite{Beyersdorff:2013:TOCL}. In the standard  model, random graphs on $n$ vertices are constructed by including every edge independently with probability $p$. {\color{black}It is known (see for example \cite{BOLLO80,Bollob_s_2001}) that $k$-cliques appear
at the threshold probability $p^*$ approximately equal to $n^{- \frac{2}{k-1}}$: If $p < p^*$, then with high probability there is no $k$-clique.} 
Following \cite{Beyersdorff:2013:TOCL,DBLP:conf/stoc/AtseriasBRLNR18,DBLP:journals/combinatorica/LauriaPRT17} we consider random graphs $G$ on $kn$ vertices where an edge is present between two vertices in distinct blocks 
with probability  $p=n^{- (1+\epsilon)\frac{2}{k-1}}$, for $\epsilon$ a constant.  We call this distribution $\mathcal{G}^n_{k,\epsilon}(p)$  and
we use the notation $G \sim \mathcal{G}^n_{k,\epsilon}(p)$ to say that $G$ is a graph drawn at random from  $\mathcal{G}^n_{k,\epsilon}(p)$.
In the next sections we explore lower bounds for  $\BinC(G)$ in Res($s$) for $s \geq 1$, when $G\sim \mathcal{G}^n_{k,\epsilon}(p)$.

\subsection{Isolating the properties of $G$}
 Let $\alpha$ be a constant such that $0<\alpha<1$. Define a set of vertices $U$ in $G$, $U \subseteq V$ to be an {\em $\alpha$-transversal} if: (1) $|U| = \alpha k$, and (2) for all $b\in [k]$,  $|V_b\cap U|\leq 1$.
 Let $B(U)\subseteq [k]$ be the set of blocks mentioned in $U$, and let $\overline{B(U)} = [k]\setminus B(U)$. 
 We say that $U$ is {\em extendable}  in a block $b \in \overline{B(U)}$ if there exists a vertex $a \in V_b$ that is a common
  neighbour of all nodes in $U$, i.e. $a \in N_c(U)$ 
 where $N_c(U)$ is the set of {\em common neighbours} of vertices in $U$: $N_c(U)=\{v \in V \;|\; v \in \bigcap_{u\in U}N(u)\}$.  
\medskip

Let $\sigma$ be a partial assignment (a restriction) to the variables of $\BinC(G)$ and $\beta$ a constant such that  $0<\beta<1$. We say  $\sigma$ is  {\em $\beta$-total} if $\sigma$ assigns precisely $\lfloor \beta \log n\rfloor $ bits in each block $b\in [k]$, \mbox{i.e.} $\lfloor \beta \log n\rfloor $ variables $\omega_{b,i}$  in each block $b$. \textcolor{black}{Note that in general we do not choose the same $\lfloor \beta \log n\rfloor $ bits in each block.}
Let $v=(i,a)$  be the $a$-th node in the $i$-th block in  $G$. 
 We say that a restriction $\sigma$ is {\em consistent} with $v$  if for all $j\in [\log n]$, $\sigma(\omega_{i,j})$ is either $a_j$ 
or not assigned.

\begin{definition} 
\label{def:extension}
Let $0<\alpha,\beta <1$.
An $\alpha$-transversal set of vertices $U$ is $\beta$-extendable, if for 
all $\beta$-total restrictions $\sigma$, there is a node $v^b$ in each block $b \in \overline{B(U)}$, such that 
$\sigma$ is consistent with $v^b$.
\end{definition}
\textcolor{black}{An $\alpha$-transversal is just a set of vertices $U$ comprised of a single vertex from each of $\alpha k$ blocks. It is $\beta$-extendable if, for any restriction assigning $\lfloor \beta \log n\rfloor$ bits in each block, there is a vertex adjacent to $U$ in each block outside of $U$.} 
 
\begin{lemma}[Extension Lemma]
\label{lem:extension}
Let $0<\epsilon <1$, let $k\leq \log n$.  Let $1>\alpha >0$ and $1>\beta>0$ such that \textcolor{black}{$1-\beta >4\alpha(1+\epsilon)$}.
Let $G \sim \mathcal{G}^n_{k,\epsilon}(p)$. \textcolor{black}{Over choices of the graph $\mathcal{G}$, with probability strictly greater than zero,} both the following properties hold:

\begin{enumerate}
\item all $\alpha$-transversal sets $U$ are $\beta$-extendable; 
\item $\mathcal{G}$ does not have a $k$-clique.
 \end{enumerate}
\label{lem:stefan}
\end{lemma}
\begin{proof}
Let $U$ be an $\alpha$-transversal set and $\sigma$ be a $\beta$-total restriction. The probability that a vertex $w$ is  in $N_c(U)$ is $p^{\alpha k}$. Hence $w 	\not \in N_c(U)$  with probability  
$(1-p^{\alpha k})$. After $\sigma$ is applied, in each block $b \in \overline{B(U)}$ \textcolor{black}{there} remain $2^{\log n - \beta \log n}=
n^{1-\beta}$ available consistent vertices. Hence the probability that we cannot extend $U$ in each block of  $\overline{B(U)}$ after $\sigma$ is applied is $(1-p^{\alpha k})^{n^{1-\beta}}$.  Fix $c=2+\textcolor{black}{2}\epsilon$ and $\delta=1-\beta- \textcolor{black}{2}\alpha c$. Notice that $\delta>0$ by our choice of $\alpha$ and $\beta$.  Since $p=\frac{1}{n^{\frac{c}{\textcolor{black}{k-1}}}}$, the previous probability is \textcolor{black}{$(1-1/n^{\alpha c (\nicefrac{k}{k-1})})^{n^{1-\beta}}$, which is at most $(1-1/n^{2\alpha c})^{n^{1-\beta}}$, which in turn is at most $e^{-\frac{n^{1-\beta}}{n^{2\alpha c}}}= e^{-n^{\delta}}$ (since $e^{-x} = \lim_{m\to\infty} (1 - x/m)^m$ and indeed $e^{-x} \geq (1 - x/m)^m$ when $x,m \geq 1$)}.

There are ${k \choose \alpha k}$ possible $\alpha$-transversal sets $U$ and  $({\log n \choose \beta \log n}\cdot 2^{\beta\log n})^k$ possible 
$\beta$-total restrictions $\sigma$. \textcolor{black}{Let us count the combinations of these:}

$$
\begin{array}{lll}
{k \choose \alpha k} \cdot ({\log n \choose \beta \log n}\cdot 2^{\beta\log n})^k & \leq k^{\alpha k} \cdot (\log n)^{\beta k \log n} \cdot 2^{\beta k \log n} \\
& \leq  2^{\alpha k \log k + \beta k \log n \log \log n + \beta k \log n} \\
& \leq 2^{\log^3 n}.
\end{array}
$$
Note that the last inequality uses $k\leq \log n$. 
Hence the probability  that there is in $G$ no $\alpha$-transversal set $U$ which is $\beta$-extendable is \textcolor{black}{at most $e^{-n^{\delta}} \cdot 2^{\log^3 n}$ which is tending to zero as 
$n$ tends to infinity}.

\medskip

To bound the probability that $\mathcal{G}$ contains a $k$-clique, notice that the expected number of  $k$-cliques \textcolor{black}{can be calculated from the potential maximal number of $k$-cliques multiplied by the probability that each of these forms a $k$-clique, that is} $\textcolor{black}{n^k} \cdot p^{{k \choose 2}} \textcolor{black}{=} n^k \cdot p^{(k(k-1)/2)}$. Recalling $p=1/n^{c/\textcolor{black}{k-1}}$, we get 
that the 
\textcolor{black}{expected number of $k$-cliques is $n^k \cdot n^{-ck/2} = n^{k-ck/2}$. Since $c=2+2\epsilon$,  
$k-ck/2= -\epsilon k$.  Hence $n^k \cdot n^{-ck/2} = n^{-\epsilon k}\leq n^{-\epsilon}$, which is tending to zero as $n$ tends to infinity.}

So the probability that either property (1) or (2) does not hold is bounded above by  \textcolor{black}{$2^{\log^{3}n}\cdot e^{-n^{\delta}} + n^{-\epsilon}$} which is \textcolor{black}{strictly less than one} for sufficiently large $n$. 
\end{proof}

\subsection{$\RES(s)$ lower bounds for $\BinC$}

Let $s\geq 1$ be an integer. Call a $\frac{1}{2^{s+1}}$-total assignment  to the variables of $\BinC(G)$ an {\em $s$-restriction}.   
A  \emph{random $s$-restriction} for $\BinC(G)$ is an $s$-restriction obtained by choosing independently in each block $i$, $ \lfloor \frac{1} {2^{s+1}} \log n\rfloor $ variables among $\omega_{i,1},\ldots,\omega_{i,\log n}$, and setting these uniformly at random to $0$ or $1$.  


\medskip
Let $s,k\in \mathbb N $, $s,k\geq1$ and let $G\textcolor{black}{ \sim \mathcal{G}^n_{k,\epsilon}(p)}$ be a graph over $nk$ nodes and $k$ blocks which does not contain a $k$-clique. Fix $\delta=\textcolor{black}{\frac{1}{2\cdot  96^2}}$ and $\pe(s)= {\color{black} 2^{s^2 + 3s}}$ and $\de(s)=(\pe(s)s)^s$.


\textcolor{black}{Let $\BinC(G)\!\!\!\upharpoonright_\rho$ denote $\BinC(G)$ restricted by $\rho$.} Consider the following property.
\begin{definition}[Property $\Cl(G, s,k)$] \label{def:cliqueProp}
	For any \textcolor{black}{$s$-restriction} $\rho$, there are no Res($s$) refutations of 
	$\BinC(G)\!\!\!\upharpoonright_\rho$ 
	of size less than $n^{\frac{\delta(k-1)}{\de(s)}}$.
\end{definition}
\noindent If property $\Cl(G,s,k)$ holds, we immediately  have a $n^{\Omega(k)}$ \sloppy size lower bound for refuting $\BinC(G)$ in $\RES(s)$.
\begin{corollary}
	\label{cor:BinC}
	Let $s,k$ be integers, $s\geq1,k>1$. Let $G$ be a graph and assume that  $\Cl(G,s,k)$ holds. 
	Then there are no Res($s$) refutations of   $\BinC(G)$ of size smaller than $n^{\delta \frac{k-1}{\de(s)}}$. 
\end{corollary}
\begin{proof}
	Choose $\rho$ to be any $s$-restriction. The result follows from the previous definition since the shortest refutation of a restricted principle can never be larger than the shortest refutation of the unrestricted principle.
\end{proof}
\noindent We use the previous corollary to prove lower bounds for $\BinC(G)$ in $\RES(s)$ as long as $s =  o((\log\log n)^\frac{1}{3})$.
\begin{theorem}
	\label{thm:main-BinC}
	Let $0<\epsilon<1$ be given.
	Let $k$ be an integer with $k > 1$, and $s$ be an integer with  \textcolor{black}{$1< s \leq \left(\frac{1}{5} \log \log n \right)^{\frac{1}{3}}$}.
	Then there exists a graph $G$ such that all Res($s$) refutations of $\BinC(G)$ have size  at least \textcolor{black}{$n^{\nicefrac{\Omega(k)}{\de(s)}}$}. 
\end{theorem}

\begin{proof}
	\textcolor{black}{Let $\beta=\frac{3}{4}$ and $\alpha=\frac{1}{16(1+2\epsilon)}$}. \textcolor{black}{Note that as $\alpha < \frac{1}{16}$, $1-\beta>4\alpha(1+\epsilon)$ holds.}
	
	By Lemma~\ref{lem:extension}, we can fix $G\sim \mathcal{G}^n_{k,\epsilon}$ such that: 
	\begin{enumerate}
		\item all $\alpha$-transversal sets $U$ are $\beta$-extendable; 
		\item $\mathcal{G}$ does not have a $k$-clique.
	\end{enumerate}
	We will prove, by induction on \textcolor{black}{$s \leq \left(\frac{1}{5} \log \log n\right)^{\frac{1}{3}}$}, that property $\Cl(G,s,k)$ does hold. \textcolor{black}{Lemma \ref{lem:Res1-bin-k-clique} is the base case and Lemma \ref{lem:ind} the inductive case. The result then follows by Corollary \ref{cor:BinC}}. 
\end{proof}


\begin{lemma}[Base Case]
	\label{lem:Res1-bin-k-clique}
	$\Cl(G,1,k)$ does hold.
	\label{cor:Res1-bin-k-clique-bis}
\end{lemma}
\begin{proof}
	Fix $\beta=\frac{3}{4}$ and $\alpha=\textcolor{black}{\frac{1}{16(1+2\epsilon)}}$. \textcolor{black}{Note that  $\frac{1}{16} > \alpha >\frac{1}{48}$ and} $\de(1)=16$. \textcolor{black}{Notice also that $1-\beta>4\alpha(1+\epsilon)$ holds.}
	
	Let $\rho$ be a $1$-restriction, that is, a $\frac{1}{4}$-total assignment. 
	We claim that any Resolution refutation of $\BinC(G)\!\!\!\upharpoonright_\rho$ must have width at least $\frac{k\log n}{{\color{black} 96}}$. This is a consequence of \textcolor{black}{Property 1 of the Extension Lemma (\ref{lem:extension}), which we henceforth abbreviate as the \emph{extension property}, which allows Adversary to play against Prover with the following strategy. For each block, while fewer than $\frac{\log n}{2}$ bits are known, Adversary offers Prover a free choice. Once $\frac{\log n}{2}$ bits are set, then Adversary chooses an assignment for the remaining bits according to the extension property. 
		Summing up the $\frac{1}{4}$ (proportion of bits in the $\frac{1}{4}$-total assignment) with a potential further $\frac{1}{2}$ of the bits set in the game, we obtain no more than  $\frac{3}{4}=\beta$ proportion of bits set, in each block (though the bits set in each block need not be the same). Using the extension property separately in each block, we can guarantee that an appropriate assignment to the remaining bits also exists. This allows the game to continue until some CNF record has width at least $\frac{\log n}{2} \cdot \frac{k}{48}=\frac{k\log n}{96}$.} Size-width tradeoffs for Resolution  \cite{Ben-sasson99shortproofs} tell us that minimal size to refute any unsatisfiable CNF $F$ is lower bounded by  $2^{\frac{(\mathit{w(\vdash F)-w(F))}^2}{16 \mathit{V(F)}}}$\footnote{According to \cite{DBLP:journals/eccc/ECCC-TR99-041} Th 8.11 }. In our case $w(F)= 2\log n $ and $\mathit{V(F)}=k\log n$, hence 
	the minimal size required is $\geq 2^{\frac{(\frac{k \log n}{96} - 2 \log n)^2}{16 k \log n}}= 
	2^{\frac{\log n (\frac{k}{96}-2)^2}{16 k}}= n^{\frac{(\frac{k}{96} -2)^2}{16 k}}$.  It is not difficult to see that $\frac{(\frac{k}{96} -2)^2}{16k}$ \textcolor{black}{$>
		\frac{(k-1)}{2\cdot 16 \cdot 96^2}$ when $k>2 \cdot 16 \cdot 96^2$. Since $\delta=\frac{1}{2 \cdot 96^2}$} and $\de(1)=16$ the result is proved.
	
\end{proof}

{\color{black} For short, let $L(s) := n^{\frac{\delta(k-1)}{\de(s)}}$ denote the size bound from \cref{def:cliqueProp}.}

\begin{lemma}[Inductive Case]
\label{lem:ind}
$\Cl(G,s-1,k) \mbox{ implies } \Cl(G,s,k).$
\end{lemma}

\begin{proof} 
{\color{black}

Assume (towards a contradiction) the opposite -- that $\Cl(G,s-1,k)$ holds but there is some $s$-restriction $\rho$ such that $\BinC(G)\!\!\!\upharpoonright_\rho$ has a refutation $\pi$ of size strictly less than $L(s)$. Fix $c$ to be such that

\[
2^{c + 2} = \frac{L(s - 1)}{L(s)}.
\]

Define $r= \frac{c}{s}$ and let us call a \emph{bottleneck} a CNF record $R$ in $\pi$ whose covering number is $\geq c$. 
Hence in such a CNF record it is always possible to find $r$   pairwise disjoint $s$-tuples of literals $T_1=(\ell^1_1,\ldots,\ell^s_1),\ldots,T_r=(\ell^1_r,\ldots,\ell^s_r)$ such that the $\bigwedge T_i$'s are among the terms of the $s$-DNF forming the CNF record $R$.

Let $\sigma$ be a \emph{random $s$-restriction} on the variables of $\BinC(G)\!\!\!\upharpoonright_\rho$. 
Let us say that   {\em $\sigma$ kills a tuple $T$} if it sets to $0$ all literals in $T$ (remember that a record $s$-CNF is the negation of a $s$-DNF)
and that {\em $T$ survives $\sigma$} otherwise,  
and let us say that {\em $\sigma$ kills $R$} if it kills at least one of the  tuples in $R$. 
Let $\Sigma_i$ be the event that $T_i$ survives $\sigma$ and $\Sigma_R$ the event that $R$ survives $\sigma$. 
We claim (postponing the proof) that 
\begin{claim}
	\label{claim:main}
	If $R$ is a bottleneck, then $\Pr[\Sigma_R] \leq (1-\frac{1}{\pe(s)})^{r}$.
\end{claim}

Consider now   the restriction $\tau= \rho \sigma$. This  is a $(s-1)$-restriction on the variables of $\BinC(G)$.  We argue that in $\pi\!\!\restriction_\tau$, with probability more than zero, there is no bottleneck.  Notice that by the union bound the probability that there exists such a bottleneck CNF record $R$ that survives in $\pi\!\!\restriction_\tau$, is bounded by
$$
\Pr[ \exists R \in \pi\!\!\restriction_\rho : \Sigma_R] \leq |\pi\!\!\restriction_\tau\!\!| \left(1-\frac{1}{\pe(s)}\right)^{r}.
$$
{\color{black} (Recall that the probabilistic aspect here comes from $\sigma$ being a random $s$-restriction.)} We claim that this probability is $<1$. 
Notice that $(1-\frac{1}{\pe(s)})^{r}\leq e^{-\frac{c}{s\pe(s)}}$ using the definition of $r$.
So to prove the claim it is sufficient to prove that $|\pi\!\!\restriction_\tau\!\!|<e^{\frac{c}{\pe(s)s}}$. As $|\pi\!\!\restriction_\tau\!\!| \leq |\pi\!\!\restriction_\rho\!\!|$ and as by assumption $|\pi\!\!\restriction_\rho\!\!|  \leq L(s)$ we can show instead that

\[
e^{\frac{c}{s \cdot \pe(s)}} > L(s)
\]
\noindent or equivalently that $e^c \geq L(s)^{s \cdot \pe(s)}$. Now, as $c$ is increasing (in $n$ - see the discussion following the conclusion of this proof) we have, for $n$ large enough,
\[
e^c > 2^{c + 2} = \frac{L(s - 1)}{L(s)}
\]

so what we will show instead is that

\begin{gather}
	L(s - 1) \geq L(s)^{s \cdot \pe(s) + 1} \\
	\Leftrightarrow n^{\frac{\delta(k - 1)}{((s-1) \cdot \pe(s - 1)))^{s - 1}}} \geq \left( n^{\frac{\delta (k - 1)}{\left(s \cdot \pe(s) \right)^s}}  \right) ^{ s \cdot \pe(s) + 1} \\
	\Leftrightarrow \frac{1}{((s-1) \cdot \pe(s - 1)))^{s - 1}} \geq  \frac{s \cdot \pe(s) + 1}{\left(s \cdot \pe(s) \right)^s} \\
	\Leftrightarrow \left( s \cdot \pe(s) \right)^s \geq \left( s \cdot \pe(s) + 1 \right) ((s-1) \cdot \pe(s - 1)))^{s - 1}.
\end{gather}

Now, as $(s \cdot \pe(s) + 1) \leq 2 s \cdot \pe(s)$ it would suffice to show that $ s \cdot \pe(s) \geq 2^{(s - 1)^{-1}} (s-1) \cdot \pe(s - 1)$. But this is clear:

\begin{gather}
	 2^{(s - 1)^{-1}} (s-1) \cdot \pe(s - 1) \leq 2s \pe(s - 1) = 2s 2^{(s - 1)^2 + 3(s - 1)} = 2s 2^{s^2 + s - 2} \\= s 2^{s^2 + s - 1} \leq s 2^{s^2 + 3s} = s \cdot \pe(s).
\end{gather}

So there exists a specific $(s-1)$-restriction $\tau$ where $\pi\!\!\restriction_\tau$ contains no bottlenecks. Therefore, by \cref{lem:covering-number}, there is a $\RES(s-1)$ refutation of size strictly less than
\[
2^{c + 2} \cdot L(s) = L(s - 1)
\]
in direct contradiction with our inductive assumption.

}

\end{proof}

Let us ponder what lower bound we have discovered. Due to the definition of $L(s)$ the proof can be carried as long as $n^{\frac{\delta}{\de(s)}}$ (where $\de(s)=(s\pe(s))^s$ and $\pe(s)= {\color{black} 2^{s^2 + 3s}}$) is non-constant, whereupon $n^{\frac{\delta(k-1)}{\de(s)}}$ grows significantly in $k$. This holds while $(s\pe(s))^s < \log n$ {\color{black} which simplifies as

\begin{gather}
   \log \log n >  s \log ( s \pe(s)) = s \log ( s 2^{s^2 + 3s}) = s \log s + s^3 + 3 s^2.
\end{gather}

Clearly this holds if $s\leq \left(\frac{1}{5} \log \log n\right)^{\frac{1}{3}}$}. Hence we can deduce the following from Corollary~\ref{cor:BinC}.
\begin{corollary}
\label{cor:BinCreal}
Let \textcolor{black}{$s\leq \left(\frac{1}{5} \log \log n\right)^{\frac{1}{3}}$} and $k\leq \log n$ be integers. Choose $G$ so that $\Cl(G,s,k)$ holds (knowing that such exists). Then there are no Res($s$) refutations of  $\BinC(G)$ of size smaller than $n^{\delta \frac{k-1}{\de(s)}}$, which is of the form $g(n)^k$ for some strictly increasing function $g$. 
\end{corollary}

\bigskip

\begin{proof}(of Claim \ref{claim:main})
Since $T_1,\ldots,T_r$ are tuples in $R$, then
$\Pr[\Sigma_R] \leq  \Pr[\Sigma_1 \wedge \ldots \wedge \Sigma_r]$. Moreover 
$\Pr[\Sigma_1 \wedge \ldots \wedge \Sigma_r]=\prod_{i=1}^r\Pr[\Sigma_i | \Sigma_1 \wedge \ldots \wedge \Sigma_{i-1}]$.
We will prove that for all $i=1,\ldots, r$,

\begin{eqnarray}
\Pr[\Sigma_i | \Sigma_1 \wedge \ldots \wedge \Sigma_{i-1}]\leq \Pr[\Sigma_i]. \label{eq:main}
\end{eqnarray} 

Hence the result follows from  Lemma \ref{lem:prefect} which is proving that  
$\Pr[\Sigma_i] \leq 1-\frac{1}{\pe(s)}$. 

By Lemma \ref{lem:probcond} (i), to prove that Equation \ref{eq:main} holds, we show that
 $\Pr[\Sigma_i | \neg \Sigma_1 \vee \ldots \vee \neg \Sigma_{i-1}] \geq \Pr[\Sigma_i]$.
We claim that  for $j \in [r], i\not = j$: 
\begin{eqnarray}
\Pr[\Sigma_i | \neg \Sigma_j] \geq \Pr[\Sigma_i] \label{Eq}
\end{eqnarray}
Hence repeated applications of Lemma \ref{lem:probcond} (ii), prove that
$\Pr[\Sigma_i | \neg \Sigma_1 \vee \ldots \vee \neg \Sigma_{i-1}] \geq \Pr[\Sigma_i].$

To prove Equation \ref{Eq}, let $B(T_i)$ be the set of blocks mentioned in $T_i$. 
If $B(T_i)$ and $B(T_j)$ are disjoint, then clearly $\Pr[\Sigma_i | \neg \Sigma_j] = \Pr[\Sigma_i]$. 
When $B(T_i)$ and $B(T_j)$ are not disjoint,  we reason as follows:
For each $\ell \in B(T_i)$, let $T_i^\ell$ be the set of variables in $T_i$ mentioning block $\ell$.
$T_i$ is hence partitioned into $\bigcup_{\ell\in B(T_i)}T^\ell_i$ and hence the event ``$T_i$ surviving  $\sigma$'', 
can be partitioned into the events that $T^\ell_i$ survives $\sigma$, for  $\ell \in B(T_i)$.
Denote by $\Sigma^\ell_i$  the event ``$T^\ell_i$ survives $\sigma$'' and
let A=$B(T_i)\cap B(T_j)$ and $B=B(T_i)\setminus (B(T_i)\cap B(T_j))$.  The following equalities hold:

\begin{eqnarray}
\Pr[\Sigma_i |\neg \Sigma_j] &=&\Pr[\exists \ell \in B(T_i): \Sigma_i^\ell | \neg \Sigma_j] \\
& =& \sum_{\ell \in B(T_i)} \Pr[\Sigma_i^\ell | \neg \Sigma_j]\\
&=&\sum_{\ell \in A} \Pr[\Sigma_i^\ell  | \neg \Sigma_j] +\sum_{\ell \in B} \Pr[\Sigma_i^\ell  | \neg \Sigma_j].\\
\end{eqnarray}

Since $B$ is disjoint from $B(T_j)$, as for the case above for each $\ell \in B$, $\Pr[\Sigma_i^\ell  | \neg \Sigma_j]= \Pr[\Sigma_i^\ell]$. Then:
\begin{eqnarray}
\sum_{\ell \in B} \Pr[\Sigma_i^\ell  | \neg \Sigma_j] = \sum_{\ell \in B} \Pr[\Sigma_i^\ell]. \\
\end{eqnarray}

Notice that  $T_i$ and $T_j$ are disjoint, hence knowing that some indices in blocks $\ell \in A$ are already chosen to 
kill $T_j$, only increase the chances of $T_i$ 
to survive (since less positions are left in the blocks  $\ell \in A$ to potentially kill $T_i$). 

Hence:
\begin{eqnarray}
\sum_{\ell \in A} \Pr[\Sigma_i^\ell  | \neg \Sigma_j] \geq \sum_{\ell \in A} \Pr[\Sigma_i^\ell]. \\
\end{eqnarray}

Which proves the claim since:
\begin{eqnarray}
\sum_{\ell \in A} \Pr[\Sigma_i^\ell ] + \sum_{\ell \in B} \Pr[\Sigma_i^\ell]= \Pr[\Sigma_i]. 
\end{eqnarray}

\end{proof}

\textcolor{black}{Let $T=(\ell_{i_1,j_1},\ldots, \ell_{i_s,j_s})$ be an $s$-tuple made of disjoint literals of 
$\BinC(G)$. We say that $T$ is {\em perfect} if all literals are bits of a same block.}
\begin{lemma}
\label{lem:prefect}
Let  $\rho$ be a $s$-random restriction \textcolor{black}{and let $s\leq \left(\frac{1}{5} \log \log n\right)^{\frac{1}{3}}$}. \textcolor{black}{Let $T$ be a perfect $s$-tuple of literals from $\BinC(G)$.}
For all $s$-tuples $S$:
\[
\Pr[\mbox{$S$ survives $\rho$}] \leq \textcolor{black}{\Pr[\mbox{$T$ survives $\rho$}],}
\]
\textcolor{black}{and so 
\[
\Pr[\mbox{$S$ survives $\rho$}] \leq  1-\frac{1}{\pe(s)} = 1-\frac{1}{2^{s^2+3s}}.
\]}
\label{lem:survival-maximised-k-clique}
\end{lemma}
\begin{proof}


Let $\gamma=\frac{1}{2^{s+1}}$. A block with $r$ distinct bits contributes a factor of
\[ \frac{{\gamma \log n \choose r}}{{\log n \choose r}} \cdot \frac{1}{2^r} \]
to the probability that the $s$-tuple \textbf{does not} survive. Expanding the left-hand part of this we obtain
\[ \frac{\gamma \log n \cdot \gamma \log n \, -1 \cdots \gamma \log n \, -r+1}{\log n \cdot \log n \, -1 \cdots \log n \, -r+1} = \gamma \frac{\log n}{\log n} \cdot \gamma \frac{\log n \, -\frac{1}{\gamma}}{\log n \, -1} \cdots \gamma \frac{\log n \, -\frac{r}{\gamma}+\frac{1}{\gamma}}{\log n \, -r+1}. \]
Next, let us note that
\[  1 = \frac{\log n}{\log n} > \frac{\log n \, -\frac{1}{\gamma}}{\log n \, -1} > \cdots > \frac{\log n \, -\frac{r}{\gamma}+\frac{1}{\gamma}}{\log n \, -r+1} \textcolor{black}{>\frac{1}{2}} \]
\textcolor{black}{while $r\leq s$. This is because $2(\log n - 2^{s+1}s+2^{s+1}) \geq \log n-s+1$ reduces to $\log n \geq 2^{s+2}s-2^{s+2}-s+1$ which holds while $s\leq \left(\frac{1}{5} \log \log n\right)^{\frac{1}{3}}$}.

\textcolor{black}{Calculating $\gamma^s\cdot \nicefrac{1}{2^s} \cdot \nicefrac{1}{2^s}= \nicefrac{1}{2^{s^2+3s}}$,}
the result now follows when we recall that the probability of surviving is maximised when the probability of not surviving is minimised.
\end{proof}
\begin{lemma}
\label{lem:probcond} Let $A,B,C$ three events such that $\Pr[A],\Pr[B],\Pr[C]>0$:
\begin{enumerate}
\item[$(i)$] If $\Pr[A|\neg B]\geq \Pr[A]$ then $\Pr[A | B]\leq \Pr[A]$,
\item[$(ii)$] 
If $\Pr[A|B] \geq \Pr[A]$ and   $\Pr[A|C] \geq \Pr[A]$, then
 $\Pr[A | B \vee C] \geq \Pr[A]$.
\end{enumerate}
\end{lemma}

\begin{proof}\color{black}
For part (i) consider the following equivalences:
$$
\begin{array}{lll}
\Pr[A]&=&\Pr[A|B]\Pr[B]+\Pr[A|\neg B]\Pr[\neg B] \\
   \Pr[A]       &=&\Pr[A|B]\Pr[B]+\Pr[A|\neg B](1-\Pr[B]) \\
   \Pr[A]       &\geq& \Pr[A|B]\Pr[B]+\Pr[A](1-\Pr[B])\\
   \Pr[A]\Pr[B]       &\geq&\Pr[A|B]\Pr[B]\\
   \Pr[A]       &\geq&\Pr[A|B] \\
   \end{array}
$$
For part (ii) consider the following inequalities: 
$$
\begin{array}{lll}
\Pr[A | B \vee C] &=& \frac{\Pr[A \wedge (B \vee C)]}{\Pr[B\vee C]} \\
 &\geq & \frac{\Pr[A \wedge B]}{\Pr[B\vee C]} +\frac{\Pr[A \wedge C]}{\Pr[B\vee C]}  \\
 & = & \frac{\Pr[A \wedge B]}{\Pr[B]} \cdot \frac{\Pr[B]}{\Pr[B\vee C]}+\frac{\Pr[A \wedge C]}{\Pr[C]} \cdot \frac{\Pr[C]}{\Pr[B\vee C]} \\
 &= & \Pr[A|B] \cdot \frac{\Pr[B]}{\Pr[B\vee C]}+\Pr[A|C] \cdot \frac{\Pr[C]}{\Pr[B\vee C]} \\
 &\geq &\Pr[A] \cdot (\frac{\Pr[B]+\Pr[C]}{\Pr[B\vee C]}) \\
 & \geq &\Pr[A].
 \end{array}
$$
\end{proof}


\section{Res(s) and the weak Pigeonhole Principle}
\label{sec:wphp}
 
For $n<m$, let $\BinPHP^m_n$ be the binary encoding of the (weak) Pigeonhole Principle. This involves variables $\omega_{i,j}$ that range over $i \in [m], j \in [\log n]$, where we assume for simplicity of our exposition that $n$ is a power of $2$. 
\textcolor{black}{
Its clauses are just $(\bigvee_{\ell=1}^{\log n} \omega^{1-a_\ell}_{i,\ell}\vee \bigvee_{\ell=1}^{\log n} \omega^{1-a_\ell}_{j,\ell})$, for $i \neq j$ and $a \in [n]$, where $\mathrm{bin}(a)$ is $a_1 \ldots  a_{\log n}$.
}
For a comparison with the unary version see Section \ref{sec:unbin}.  First notice that an analog of 
Lemma \ref{lem:reslog} holds for the Pigeonhole Principle too.
\begin{lemma}
\label{lem:php-un-bin}
Suppose there are Resolution refutations of $\pPHP^m_n$ of  size $S$. 
Then there are $\RES(\log n)$ refutations of $\BinPHP^m_n$  of size $S$.
\end{lemma}

Let $\rho$ be a partial assignment (a restriction) to the variables of $\BinPHP^m_n$. We call $\rho$ a {\em $t$-bit} restriction if $\rho$ assigns $t$ bits of each pigeon $b\in [m]$, i.e. $t$ variables $\omega_{b,i}$  for each pigeon $b$. Let $v=(i,a)$  be an assignment meaning that pigeon $i$ is assigned to hole $a$ and let
$a_1\dots a_{\log n}$ be the binary representation of $a$.  We say that a restriction $\rho$ is {\em consistent} with $v$  if for all $j\in [\log n]$, $\sigma(\omega_{i,j})$ is either $a_j$ or not assigned. We denote by $\BinPHP^m_n\!\!\!\upharpoonright_\rho$, $\BinPHP^m_n$ restricted by $\rho$. We will also consider the situation in which an $s$-bit restriction is applied to some $\BinPHP^m_n\!\!\!\upharpoonright_\rho$, creating $\BinPHP^m_n\!\!\!\upharpoonright_\tau$, where $\tau$ is an $s+t$-bit restriction.

Throughout this section, let \textcolor{black}{$u=u(n,t):=2((\log n) - t)$ and $u':=(\log n) - t$. We do not use these shorthands universally}, but sometimes where otherwise the notation would look cluttered. We also occasionally write $(\log n) - t$ as $\log n\, - t$ (note the extra space).  \textcolor{black}{We say that a pigeon is \emph{mentioned} in a CNF if some literal involving that pigeon appears in the CNF.}
 
\begin{lemma}
\label{lem:unwritten}
Let $\rho$ be a $t$-bit restriction for $\BinPHP^m_n$. Any decision DAG for $\BinPHP^m_n\!\!\!\upharpoonright_\rho$ must contain a \textcolor{black}{1-CNF} record which mentions $\frac{n}{2^{t}}$ pigeons.
\end{lemma}
\begin{proof}
Let Adversary play in the following fashion. While some pigeon is not mentioned \textcolor{black}{in the current record}, let him give Prover a free choice to answer any one of its bits as true or false. Once a pigeon is mentioned once, then let Adversary choose a hole for that pigeon by choosing some assignment for the remaining unset bits (we will later need to prove this is always possible). Whenever another bit of an already mentioned pigeon is queried, then Adversary will answer consistently with the hole he has chosen for it. Only once all of a pigeon's bits are forgotten (not including those set by $\rho$), will Adversary forget the hole he assigned it.

It remains to argue that Adversary must force Prover to produce a \textcolor{black}{1-CNF} record \textcolor{black}{mentioning at least $\frac{n}{2^{t}}$ pigeons} and for this it suffices to argue that Adversary can remain consistent with $\BinPHP^m_n\!\!\!\upharpoonright_\rho$ up until the point that such a \textcolor{black}{1-CNF} record exists. For that it is enough to show that there is always a hole available for a pigeon for which Adversary gave its only currently questioned bit as a free choice (but for which $\rho$ has already assigned some bits).

The current \textcolor{black}{1-CNF} record is assumed to have fewer than $\frac{n}{2^{t}}$ literals and therefore must mention fewer than $\frac{n}{2^{t}}$ pigeons, each of which Adversary already assigned a hole. Each hitherto unmentioned pigeon that has just been given a free choice has $\log n \ - t$ bits which corresponds to $\frac{n}{2^t}$ holes. Since we have assigned fewer than $\frac{n}{2^t}$ pigeons  to holes, one of these must be available, and the result follows. 
\end{proof}

\noindent Let $\xi(s)$ satisfy $\xi(1)=1$ and $\xi(s)=\xi(s-1)+1+s$. Note that $\xi(s)=\Theta(s^2)$.


\begin{definition}[Property $\pPHP(s,t)$]
Let $s,t\geq 1$. For any $t$-bit restriction $\rho$ to $\BinPHP^m_n$, there are no $\RES(s)$ refutations of $\BinPHP^m_n\!\!\!\upharpoonright_\rho$ of size smaller than $e^{\frac{n}{4^{\xi(s)+1} s! 2^t u^{\xi(s)}}}\textcolor{black}{=\exp (\frac{n}{4^{\xi(s)+1} s! 2^t u^{\xi(s)}})}$.
\end{definition}

\begin{theorem}
\label{thm:reswphp}
Let $\rho$ be a $t$-bit restriction for $\BinPHP^m_n$. Any decision DAG for $\BinPHP^m_n\!\!\!\upharpoonright_\rho$ is
\textcolor{black}{of size $\geq e^{\frac{n}{2^{t+1} u}}$ (which is $2^{\Omega(\frac{n}{\log n})}$ at $t=0$).}
\end{theorem}
\begin{proof}
Call a \emph{bottleneck} a \textcolor{black}{1-CNF} record in the decision DAG that mentions $\frac{n}{2^{t+1}}$ pigeons. Now consider a random restriction that picks for each pigeon one bit uniformly at random and sets this to $0$ or $1$ with equal probability. The probability that a bottleneck survives (is not falsified by) the random restriction is no more than 
\[\textcolor{black}{\left( \frac{u'-1}{u'} + \frac{1}{2 u'} \right)^{\frac{n}{2^{t+1}}} = \left( 1 - \frac{1}{2 u'} \right)^{u' \cdot \frac{n}{2^{t+1} u'}} = \left( 1 - \frac{1}{u} \right)^{u \cdot \frac{n}{2^{t+1} u}}\leq \frac{1}{e^{\frac{n}{2^{t+1} u}}},} \]
since $e^{-x} = \lim_{m\to\infty} (1 - x/m)^m$ and indeed $e^{-x} \geq (1 - x/m)^m$ when $x,m \geq 1$.

Now suppose for contradiction that we have fewer than \textcolor{black}{$e^{\frac{n}{2^{t+1} u}}$} bottlenecks in a decision DAG for $\BinPHP^m_n\!\!\!\upharpoonright_\rho$. By the union bound there is a random restriction that kills all bottlenecks and this leaves a decision DAG for some $\BinPHP^m_n\!\!\!\upharpoonright_\sigma$, where $\sigma$ is a $(t+1)$-bit restriction for $\BinPHP^m_n$.
However, we know from Lemma~\ref{lem:unwritten} that such a refutation must involve a \textcolor{black}{1-CNF} record mentioning $\frac{n}{2^{t+1}}$ pigeons. This is now the desired contradiction.
\end{proof}
\textcolor{black}{While $m$ is linear in $n$}, the previous theorem could have been proved, like Lemma~\ref{lem:Res1-bin-k-clique}, by the size-width trade-off. However, the method of random restrictions used here could not be easily applied there, due to the randomness of $G$.
\begin{corollary}
Property $\pPHP(1,t)$ holds, for each $t<\log n$.
\end{corollary}
\noindent Note that, $\pPHP(1,t)$ yields only trivial bounds as $t$ approaches $\log n$.

Let $(\ell_{i_1,j_1},\ldots, \ell_{i_s,j_s})$ be an $s$-tuple made of disjoint literals of 
$\BinPHP^m_n\upharpoonright_\rho$. We say that a tuple is \textcolor{black}{{\em anti-perfect}} if all literals come from \textcolor{black}{different pigeons}.

\begin{lemma}
Let $s$ be an integer, $s \geq 1$ and $s+t<\log n$. Let $\sigma$ be a random $s$-bit restriction over $\BinPHP^m_n \!\!\!\upharpoonright_\rho$ where $\rho$ is itself some $t$-bit restriction over $\BinPHP^m_n$. Let $T$ be \textcolor{black}{an anti-perfect} $s$-tuple of $\BinPHP^m_n\!\!\!\upharpoonright_\rho$. Then for all $s$-tuples S:
$$\Pr[\mbox{$T$ survives $\sigma$}] \geq \Pr[\mbox{$S$ survives $\sigma$}].$$
and so $\Pr[\mbox{$S$ survives $\sigma$}] \leq 1 - \frac{1}{u^s}$.
\label{lem:survival-maximised-PHP}
\end{lemma}
\begin{proof}
A pigeon with $r$ distinct bits contributes to not surviving a factor of
\[ \frac{s }{\log n \, -t} \cdot \frac{s-1}{\log n \, -t -1} \cdots \frac{s-r+1}{\log n -t -r+1} \cdot \frac{1}{2^r}.\]
Noting that
\[ \frac{s }{\log n \, -t} \cdot \frac{s-1}{\log n \, -t -1} \cdots \frac{s-r+1}{\log n -t -r+1} \cdot \frac{1}{2^r} > \textcolor{black}{\frac{1}{(2u')^r} = \frac{1}{u^r}} \]
the result now follows when we recall that the probability of surviving is maximised when the probability of not surviving is minimised.
\end{proof}

\begin{theorem}
Let $s>1$ and $s+t<\log n$. Then, $\pPHP(s-1,s+t)$ implies $\pPHP(s,t)$.
\label{thm:new-wphp-main}
\end{theorem}
\begin{proof}
We proceed by contraposition. Assume there is some $t$-bit restriction $\rho$ so that there exists a $\RES(s)$ refutation $\pi$ of $\BinPHP^m_n\!\!\!\upharpoonright_\rho$ with size less than $e^{\frac{n}{4^{\xi(s)+1} \cdot s! 2^t u^{\xi(s)}}}\textcolor{black}{=\exp(\frac{n}{4^{\xi(s)+1} \cdot s! 2^t u^{\xi(s)}})}$.

Call a \emph{bottleneck} a CNF record that has covering number $\geq \frac{n}{4^{\xi(s)} \cdot (s-1)! 2^t u^{\xi(s-1)}}$.  In such a CNF record, by dividing by $s$ and $u$, it is always possible to find $r:=\frac{n}{4^{\xi(s)} s! 2^t u^{\xi(s-1)+1} }$ $s$-tuples \sloppy of literals $(\ell^1_1,\ldots,\ell^s_1),\ldots,(\ell^1_r,\ldots,\ell^s_r)$ so that each $s$-tuple is a clause in the CNF record and no pigeon appearing in the $i$th $s$-tuple also appears in the $j$th $s$-tuple (when $i \neq j$). This important independence condition plays a key role. Now consider a random restriction that, for each pigeon, picks uniformly at random $s$ bit positions and sets these to $0$ or $1$ with equal probability. The probability that the $i$th of the $r$ $s$-tuples survives the restriction is maximised when each \textcolor{black}{each variable among the $s$ describes} a different pigeon (by Lemma~\ref{lem:survival-maximised-PHP}) and is therefore \textcolor{black}{bounded} above by
\[ \left( 1 - \frac{1}{u^s} \right) \]
whereupon
\[ \left( 1 - \frac{1}{u^s} \right)^{\frac{n}{4^{\xi(s)} s! 2^t u^{\xi(s-1)+1}}} = \left( 1 - \frac{1}{u^s} \right)^{\frac{n u^s}{4^{\xi(s)}s! 2^t u^{(\xi(s-1)+1+s)}}} \]
which is $\textcolor{black}{\leq 1/e^{\frac{n}{4^{\xi(s)}s! \cdot 2^{t} u^{\xi(s)}}}} \leq 1/e^{\frac{n}{4^{\xi(s)+1}s! \cdot 2^{t} u^{\xi(s)}}}$.
Supposing therefore that there are fewer than $e^{\frac{n}{4^{\xi(s)+1}s! \cdot 2^{t} u^{\xi(s)}}}$ bottlenecks, one can deduce a random restriction that kills all bottlenecks. What remains after doing this is a $\RES(s)$ refutation of some $\BinPHP^m_n\!\!\!\upharpoonright_\sigma$, where $\sigma$ is a $s+t$-bit restriction, which moreover has covering number $< \frac{n}{4^{\xi(s)} \cdot (s-1)! 2^t u^{\xi(s-1)}}$. But if the remaining $\RES(s)$ refutation is of size $<e^{\frac{n}{4^{\xi(s)+1}s! \cdot 2^{t} u^{\xi(s)}}}$ then,  from Lemma~\ref{lem:covering-number}, it would give a $\RES(s-1)$ refutation of size
\[ <e^{\frac{n}{4^{\xi(s)} \cdot (s-1)! 2^t u^{\xi(s-1)}}} \cdot e^{\frac{n}{4^{\xi(s)+1}s! \cdot 2^{t} u^{\xi(s)}}} = e^{\frac{n}{4^{\xi(s)} \cdot (s-1)! 2^t u^{\xi(s-1)}} (1 + \frac{1}{4 s u^{s+1}})} \]
\[< e^{\frac{2n}{4^{\xi(s)} \cdot (s-1)! 2^t u^{\xi(s-1)}}} < e^{\frac{n}{4^{\xi(s)} \cdot (s-1)! 2^{t-1} u^{\xi(s-1)}}} <  
e^{\frac{n}{4^{\xi(s)-s} \cdot (s-1)! 2^{s+t} u^{\xi(s-1)} } },\]
since $4^{s} > 2^{s+1}$, which equals $e^{\frac{n}{4^{\xi(s-1)+1} \cdot (s-1)! 2^{s+t} u^{\xi(s-1)}}}$ in contradiction to the inductive hypothesis.
\end{proof}

\begin{theorem}
\label{thm:wphp}
Fix $\lambda,\mu>0$. Any refutation of $\BinPHP^m_n$ in $\RES(\sqrt{2}\log^{\frac{1}{2}-\lambda}n)$ is of size $2^{\Omega(n^{1-\mu})}$.
\end{theorem}
\begin{proof}
First, let us claim that $\pPHP(\sqrt{2}\log^{\frac{1}{2}-\lambda}n,0)$ holds (and this would hold also at $\lambda=0$). \textcolor{black}{Repeated application of Theorem~\ref{thm:new-wphp-main} gives $s$ such that $\sum_{i=1}^s i = \frac{s(s+1)}{2}<\log n$. Noting $\frac{s^2}{2}<\frac{s(s+1)}{2}$, the claim follows}.

Now let us look at the bound we obtain by plugging in to $e^{\frac{n}{4^{\xi(s)+1} \cdot s! 2^t u^{\xi(s)}}}$ at $s=\sqrt{2} \log^{\frac{1}{2} - \lambda}n$ and $t=0$. We recall $\xi(s)=\Theta(s^2)$. Note that, when $\lambda>0$, each of $4^{\xi(s)+1}$, $s!$ and $\log^{\xi(s)}n$ is $o(n^{\mu})$. The result follows.
\end{proof}

\subsection{The treelike case}
\label{subsec:treephp}
Concerning the Pigeonhole Principle,  we can prove that the relationship between $\PHP^{n+1}_n$ and $\BinPHP^{n+1}_n$ is different for treelike Resolution from general Resolution. In particular, for very weak Pigeonhole Principles, we know the binary encoding is harder to refute in general Resolution; whereas for treelike Resolution it is the unary encoding which is the harder.

\begin{theorem}
The treelike Resolution complexity of $\BinPHP^m_n$ is $2^{\Theta(n)}$.
\label{thm:phpubtreelike}
\end{theorem}
\begin{proof}
For the lower bound, one can follow the proof of Lemma~\ref{lem:unwritten} with $t=0$ and find $n$ free choices on each branch of the tree. Following the method of Riis \cite{SorenGap}, we uncover a subtree of the decision tree of size $2^n$.

For an upper bound of $2^{2n}$ we pursue the following strategy. First we choose some $n+1$ pigeons to question. We then question all of them on their first bit and separate these into two sets $T_1$ and $F_1$ according to whether this was answered true or false. If $n$ is a power of $2$, choose the larger of these two sets (if they are the same size then choose either). If $n$ is not a power of two, the matter is mildly complicated, and one must look at how many holes are available with the first bit set to $1$, say $h^1_1$; versus $0$, say $h^0_1$. At least one of $|T_1|>h^1_1$ or $|F_1|>h^0_1$ must hold and one can choose between $T_1$ and $F_1$ correspondingly.
Now question the second bit, producing two sets $T_2$ and $F_2$, and iterate this argument. We will reach a contradiction in $\log n$ iterations since we always choose a set of maximal size. The depth of our tree is bound above by $n+\frac{n}{2} + \frac{n}{4} + \cdots < 2n$ and the result follows. 
\end{proof}

\section{The \SA\ size lower bound for the binary Pigeonhole Principle}
\label{sec:SAbinPHP}

In this section we study the inequalities derived from the binary encoding of the Pigeonhole principle, whose axioms we remind the reader of now. {\color{black} $\BinPHP_{n}^{m}$ has, for each two distinct pigeons $ i \not = i' \in [m]$ and each hole $a \in [n]$, the axiom $\sum_{j=1}^{\log n}  \omega_{i,j}^{(1 - a_j)}  +  \sum_{j=1}^{\log n} \omega_{i',j}^{(1 - a_j)} \geq 1$, where $a_1 \ldots a_{\log n}$ is the binary representation of $a$}. We first prove a certain \SA\ rank lower bound for a version of
the binary \PHP, in which only a subset of the holes is available.
\begin{lemma}
\label{lemma:degree-bound}Let $H\subseteq\left[n\right]$ be a subset
of the holes and let us consider  $\BinPHP_{|H|}^{m}$ where each
pigeon can go to a hole in $H$ only. Any \SA\ refutation of $\BinPHP_{|H|}^{m}$
involves a term that mentions at least $\left|H\right|$ pigeons.
\end{lemma}
\begin{proof}
We get a valuation $v$ from a partial matching in an obvious way. \textcolor{black}{That is, if a pigeon $i$ is assigned to hole $a$, whose representation in binary is $a_1 \ldots a_{\log n}$, then we set each $\omega_{i,j}^{a_j}$ to $a_j$.}
We say that a product term $P=\text{\ensuremath{\prod_{j \in J} \omega_{i_{j},k_{j}}^{b_{j}}}}$
mentions the set of pigeons $M = \left\{ i_{j} : j \in J\right\} $. Let us denote
the number of available holes by $n' := |H|$. Every product term that mentions at
most $n'$ pigeons is assigned a value $v\left(P\right)$ as follows.
The set of pigeons mentioned in $M$ is first extended arbitrarily to a
set $M'$ of exactly $n'$ pigeons. $v\left(P\right)$ is then the
probability that a matching between $M'$ and $H$ taken uniformly
at random is consistent with the product term $P$. In other words, $v\left(P\right)$
is the number of perfect matchings between $M'$ and $H$ that are
consistent with $P$, divided by the total, $(n')$!. Obviously, this
value does not depend on how $M$ is extended to $M'$. Also, it is
symmetric, i.e. if $\pi$ is a permutation of the pigeons, $v\left(\ensuremath{\prod \omega_{i_{j},k_{j}}^{b_{j}}}\right)=\ensuremath{v\left(\prod \omega_{\pi\left(i_{j}\right),k_{j}}^{b_{j}}\right)}$.

All lifts of axioms of equality $\omega_{j,k} + \textcolor{black}{\neg}\omega_{j,k} = 1$ are automatically satisfied since a
matching consistent with $P$ is consistent either with $P\omega_{j,k}^{b}$
or with $P\omega_{j,k}^{1-b}$ but not with both, and thus
\[
v\left(P\right)=v\left(P\omega_{j,k}^{b}\right)+v\left(P\omega_{j,k}^{1-b}\right).
\]
Regarding the lifts of the disequality of two pigeons $i \neq j$ in one hole, that is, the inequalities coming from the only clauses in $\BinPHP^m_{|H|}$, it
is enough to observe that it is consistent with any perfect matching,
\mbox{i.e.} at least one variable on the LHS is one under such a matching.
Thus, for a product term $P$, any perfect matching consistent with $P$ will
also be consistent with $P\omega_{i,k}^{1-b_{k}}$ or with $P\omega_{j,k}^{1-b_{k}}$
for some $k$.
\end{proof}



\subsection{The ordinary Pigeonhole Principle}

The proof of the size lower bound for the  $\BinPHP_{n}^{n+1}$
is then by a standard random restriction argument combined with the
rank lower bound above. Assume, without loss of generality, that $n$ is a power
of two. For the random restrictions $\mathcal{R}$, we consider the
pigeons one by one and with probability $\nicefrac{1}{4}$ we assign
the pigeon uniformly at random to one of the holes still available.
We first need to show that the restriction is ``good" with high probability, \mbox{i.e.} neither too big nor too small.
The former is needed so that in the restricted version we have a good lower bound,
while the latter will be needed to show that a good restriction coincides well with any reasonably big term, in the sense that they have in common a sufficiency of pigeons.

We will make use of the following version of the Chernoff Bound as appears in \cite{ChernoffSource}.
\begin{lemma}[Theorem 4.4 in \cite{ChernoffSource}]
Let $X_{1},X_{2},\dots, X_{n}$ be independent 0/1 random variables with $\mbox{Pr}\left[X_{i}=1\right]=p_{i}$.
Let $X=\sum_{i=1}^{n}X_{i}$ and $\mu=E\left[X\right]$. Then, for
every $\delta$, $0<\delta\leq1$, the following bound holds\[
\mbox{Pr}\left[X\geq\left(1+\delta\right)\mu\right]\leq e^{\frac{-\mu\delta^{2}}{3}}.\]
\label{lem:ch}
\end{lemma}

		\begin{lemma}
			\label{fact:good-restrictions}
			If $\mathcal{\left|R\right|}$ is the
			number of pigeons (or holes) assigned by $\mathcal{R}$, the probability that {\color{black}  $|\mathcal{R} | > \frac{3(n+1)}{8}$} is at
			most {\color{black} $e^{-\frac{(n+1)}{48}}$}.
		\end{lemma}
		\begin{proof}We use the Chernoff Bound from Lemma~\ref{lem:ch}. We have $p_{i}=\frac{1}{4}$ (and thus
			{\color{black} $\mu=\frac{n+1}{4}$}) and $\delta=\frac{1}{2}$. Thus, the probability {\color{black} the restriction assigns more than $\frac{3(n+1)}{8}$ pigeons to holes} is at most {\color{black} $e^{-\nicefrac{(n+1)}{48}}$}.
		\end{proof}

		\noindent We first prove that any given wide product term, \mbox{i.e.} a term that mentions a
		constant fraction of the pigeons, survives the random restrictions
		with exponentially small probability.
		\begin{lemma}
			\label{lemma:restriction-kill}Let $P$ be a product term that mentions at least
			{\color{black} $\frac{n+1}{2}$} pigeons. The probability that $P$ does not evaluate
			to zero under the random restrictions is at most  {\color{black} $\left(\frac{5}{6}\right)^{\nicefrac{n}{16}}$ (for $n$ large enough)}.
		\end{lemma}
		
		\begin{proof}
			We will desire {\color{black} $|\mathcal{R}|\leq  \frac{3(n+1)}{8}$} to ensure that at least {\color{black} $\frac{5(n+1)}{8}$} holes remain unused in $\mathcal{R}$ {\color{black} (for $n$ large enough)}. This will involve the probability {\color{black} $e^{-\nicefrac{(n+1)}{48}}$} from Lemma~\ref{fact:good-restrictions}.
			
			A further application of the Chernoff Bound from Lemma~\ref{lem:ch} ( {\color{black} $\mu=\frac{n+1}{8}$}, $\delta=-\frac{1}{2}$) gives the probability that fewer
			than {\color{black} $\frac{n+1}{16}$} pigeons mentioned by $P$ are assigned by $\mathcal{R}$
			is at most {\color{black} $e^{-\nicefrac{(n+1)}{96}}$}.
			
			For each of these {\color{black} assigned} pigeons the
			probability that a single bit-variable in $P$ belonging to the pigeon
			is set by $\mathcal{R}$ to zero is at least {\color{black} $\frac{1}{5}$}. This
			is because when $\mathcal{R}$ sets the pigeon, and thus the bit-variable,
			there were at least {\color{black} $\frac{5(n+1)}{8}$} holes available, while at most
			{\color{black} $\frac{n+1}{2}$} choices set the bit-variable to one. The difference {\color{black} -- which will be a lower bound on the number of holes available setting the selected bit to 0 --} is {\color{black} $\frac{n+1}{8}$} which when divided by {\color{black} $\frac{5(n+1)}{8}$} {\color{black} (to normalise the probability)} gives {\color{black} $\frac{1}{5}$}. Thus $P$ survives
			under $\mathcal{R}$ with probability at most {\color{black} 
				$e^{-\nicefrac{(n+1)}{48}} + e^{-\nicefrac{(n+1)}{96}}+\left(\frac{4}{5}\right)^{\nicefrac{n+1}{16}} < \left(\frac{5}{6}\right)^{\nicefrac{n}{16}}$ }.
		\end{proof}
		Finally, we can prove that
		\begin{theorem}
			\label{theorem:size-bound}Any \SA\ refutation of the  $\BinPHP_{n}^{n+1}$
			has to contain at least {\color{black} $\left(\frac{7}{6}\right)^{\nicefrac{n}{16}}$}
			terms.
		\end{theorem}
		\begin{proof}
			Assume for a contradiction, that there is a smaller refutation. We wish to argue that there is a random restriction with {\color{black} $|\mathcal{R}|\leq \frac{3(n+1)}{8}$} that evaluates to zero all terms that mention at least {\color{black} $\frac{n+1}{2}$} pigeons. There are at most {\color{black} $\left(\frac{7}{6}\right)^{\nicefrac{n}{16}}$} such terms so an application of the union-bound together with Lemma~\ref{fact:good-restrictions} {\color{black} and Lemma~\ref{lemma:restriction-kill}} gives a probability
			\[ {\color{black} \left(\frac{5}{6}\right)^{\nicefrac{n}{16}}} \times {\color{black} \left(\frac{7}{6}\right)^{\nicefrac{n}{16}}} + {\color{black} e^{-(n+1)/48}} < 1.\]
			Now we apply the random restriction which we know must exist to leave no terms mentioning at least {\color{black} $\frac{n+1}{2}$} pigeons
			in an \SA\ refutation of the  binary $\text{PHP}_{n'}^{m'}$, where $m'>n'\geq  {\color{black} \frac{5(n+1)}{8}}$. However, since $n'> {\color{black} \frac{n+1}{2}}$, this contradicts Lemma \ref{lemma:degree-bound}.
		\end{proof}
\begin{corollary}
\label{cor:size-bound} {\color{black} Any \SA\ refutation of the  $\BinPHP_{n}^{n+1}$
must have size $2^{\Theta(n)}$.}
\end{corollary}
\begin{proof}
The size lower bound comes from the previous theorem. We know that there is a  $2^{\Theta(n)}$ upper bound in treelike Resolution from Theorem~\ref{thm:phpubtreelike} and the result follows from the standard simulation of Resolution by \SA\ which increases refutations by no more than a factor which is a polynomial in $n$ \cite{TCS2009}.
\end{proof}

\subsection{The weak Pigeonhole Principle}

We now consider the so-called weak  binary PHP, $\BinPHP_{n}^{m}$, where
$m$ is potentially much larger than $n$. The weak unary $\PHP_{n}^{m}$ is interesting because it admits (significantly) subexponential-in-$n$ refutations in Resolution when $m$ is sufficiently large \cite{BussP97}. It follows that this size upper bound is mirrored in \SA. However, as proved earlier in this article the weak binary $\BinPHP_{n}^{m}$ remains almost-exponential-in-$n$ for minimal refutations in Resolution. We will see here that the weak binary $\BinPHP_{n}^{m}$ remains almost-exponential-in-$n$ for minimally sized refutations in \SA. In this weak binary case, the random restrictions
$\mathcal{R}$ above do not work, so we apply quite different restrictions
$\mathcal{R}'$ that are as follows: for each pigeon select independently
a single bit uniformly at random and set it to $0$ or $1$ with probability
of $\nicefrac{1}{2}$ each.

We can easily prove the following
\begin{lemma}
\label{lemma:restriction-kill-weak}
A product term $P$ that mentions $n'$
pigeons does not evaluate to zero under $\mathcal{R}'$ with probability
at most $e^{-\nicefrac{n'}{2\log n}}$.
\end{lemma}
\begin{proof}
For each pigeon mentioned in $P$, the probability that the bit-variable present
in $P$ is set by the random restriction is $\frac{1}{\log n}$, and
if so, the probability that the bit-variable evaluates to zero is
$\frac{1}{2}$. Since this happens independently for all $n'$ mentioned
pigeons, the probability that they all survive is at most $\left(1-\frac{1}{2\log n}\right)^{n'}$.
\end{proof}

\begin{lemma}
\label{lemma:degree-bound-weak}The probability that $\mathcal{R}'$ fails to have, for each $k \in [\log n]$ and $b \in \{0,1\}$, at least $\frac{m}{4\log n}$ pigeons with the $k$th bit set to $b$, is at most $e^{-\nicefrac{n}{48\log n}}$.
\end{lemma}
\begin{proof}
We apply the Chernoff Bound of Lemma~\ref{lem:ch} to deduce that for each bit position
$k$, $1\leq k\leq (\log n)$ and a value $b$, $0$ or $1$, the probability
that there are fewer than $\frac{m}{4\log n}$ pigeons for which the
$k$th bit is set to $b$ is at most $e^{-\nicefrac{m}{24\log n}}$. This uses $\mu=\frac{m}{2\log n}$ and $\delta=-\frac{1}{2}$.
Since $m>n$, by the union bound, the probability that this holds for some position
$k$ and some value $b$ is at most $(2\log n) e^{-\nicefrac{m}{24\log n}}\leq e^{-\nicefrac{n}{48\log n}}$.
\end{proof}
In order to conclude our result, we will profit from a graph-theoretic treatment of Hall's Marriage Theorem \cite{HallsMarriage}.
Suppose $G$ is a finite bipartite graph with bipartitions $X$ and $Y$, then an $X$-saturating matching is a matching which covers every vertex in $X$. For a subset $W$ of $X$, let $N_{G}(W)$ denote the neighborhood of $W$ in $G$, \mbox{i.e.} the set of all vertices in $Y$ adjacent to some element of $W$. 
\begin{theorem}[\cite{HallsMarriage} (see Theorem 5.1 in \cite{Lint1992ACI})]
Let $G$ be a finite bipartite graph with bipartitions $X$ and $Y$.  There is an $X$-saturating matching if and only if for every subset $W$ of $X$, $|W|\leq |N_{G}(W)|$.
\label{thm:Hall}
\end{theorem}
\begin{corollary}
Any \SA\ refutation of the  $\BinPHP_{n}^{m}$, $m > n$, has to
contain at least $e^{\nicefrac{n}{32\log^{2}n}}$ terms.
\end{corollary}

\begin{proof}
Assume for a contradiction, that there is a refutation with fewer than $e^{\nicefrac{n}{32\log^{2}n}}$ product terms. We want to argue that there is a random restriction that evaluates all terms that mention at least $\frac{n}{4\log n}$ pigeons to zero while satisfying the condition of Lemma~\ref{lemma:degree-bound-weak}. Using {\color{black} a union bound and \Cref{lemma:restriction-kill-weak} we upper bound the probability this fails to happen as} $e^{-n/8\log^2 n}\cdot e^{\nicefrac{n}{32\log^{2}n}}+e^{-\nicefrac{n}{48\log n}}<1$ so such a random restriction $\mathcal{R}'$ does exist.

Then, $\mathcal{R}'$
leaves at least $\frac{m}{4\log n}$  pigeons of each type $\left(k,b\right)$,
\mbox{i.e.} the $k$th bit of the pigeon is set to $b$. Recalling $m\geq n$, we now pick a set
of pigeons $S$ that has $(*)$ precisely $\frac{n}{4\log n}$ pigeons of
each type and thus is of size $\nicefrac{n}{2}$.

We will give an evaluation of the restricted principle which contradicts that the original object was a refutation.
We evaluate any product term $P$ that mentions at most $\frac{n}{4\log n}$
pigeons by first \textcolor{black}{relabeling the mentioned pigeons, injectively, using the labels of pigeons in $S$ while preserving types}, which we can do due to property $(*)$, and then giving it a value
as before. That is, by taking the probability that a perfect matching
between $S$ and some set of $\nicefrac{n}{2}$ holes consistent with the random restriction, is consistent
with $P$.

To finish the proof, we need to show that such a set of $\nicefrac{n}{2}$ holes exists, that is, such a matching exists. But this follows trivially from Theorem~\ref{thm:Hall} as every pigeon has $\nicefrac{n}{2}$ holes available, so at least the same applies to any set of pigeons.
\end{proof}

\section{The \SA\ rank upper bound for Ordering Principle with equality}
\label{sec:SAeqLOP}

Let us remind ourselves of the Ordering Principle in both unary and binary.

\noindent\begin{minipage}{.45\linewidth}
\color{black}
	\begin{gather*}
	\OP: \mbox{\em\underline{Unary encoding}}\\
	\neg v_{i,i} \qquad \forall i \in [n] \\
	\neg v_{i,j} \vee \neg v_{j,k} \vee v_{i,k} \qquad \forall i,j,k \in [n] \\
	\neg w_{i,j} \vee v_{i,j}  \qquad \forall i,j \in [n]  \\
	\textstyle \bigvee_{i\in [n]} w_{i, j} \qquad \forall j \in [n] \\
	\end{gather*}
\end{minipage}%
\begin{minipage}{.45\linewidth}
\color{black}
	\begin{gather*}
	\bOP: \mbox{\em \underline{Binary encoding}} \\
	\neg \nu_{i,i} \qquad \forall i \in [n] \\
	\neg \nu_{i,j} \vee \neg \nu_{j,k} \vee \nu_{i,k} \qquad \forall i,j,k \in [n] \\
	\textstyle \bigvee_{i\in [\log n]} \omega^{1-a_i}_{i,j} \vee \nu_{j,a} \qquad \forall j, a \in [n]\\
	\mbox{where $a_1\ldots a_{\log n} = \mathrm{bin}(a)$} \\
	\end{gather*}
\end{minipage} \vskip\baselineskip 
\noindent Note that we placed the witness in the variables $w_{i,x}$ as the first argument and not the second, as we had in the introduction. This is to be consistent with the $v_{i,j}$ and the standard formulation of $\mathrm{OP}$ as the least, and not greatest, number principle. A more traditional form of the (unary encoding of the) $\OP$ has clauses $\bigvee_{i\in [n]} v_{i, j}$ which are consequent on $\bigvee_{i\in [n]} w_{i, j}$ and $\textcolor{black}{\neg w_{i,j}} \vee v_{i,j}$ (for all $i \in [n]$).

In \SA, we wish to discuss the encoding of the Ordering Principle (and Pigeonhole Principle) as ILPs \emph{with equality}. For this, we take the unary encoding but instead of translating the wide clauses (\mbox{e.g.} from the $\mathrm{OP}$) from  $\bigvee_{i\in [n]} w_{i,x}$ to $w_{1,x} + \ldots + w_{n,x} \geq 1$, we instead use $w_{1,x} + \ldots + w_{n,x} = 1$. This makes the constraint at-least-one into exactly-one (which is a priori enforced in the binary encoding).  A reader \textcolor{black}{favouring a specific example may} consider the Ordering Principle as the combinatorial principle of the following lemma.
\begin{lemma} \label{lemma:binconv}
Let $\mathrm{C}$ be any combinatorial principle expressible as a first order formula in  $\Pi_2$-form with no finite models.
Suppose the unary encoding of $\mathrm{C}$ with equalities has an \SA\ refutation of rank $r$ and size $s$. Then the binary encoding of $\mathrm{C}$ has an \SA\ refutation of rank at most $r \log n$ and size at most $s$.
\end{lemma}
\begin{proof}
We take the \SA\ refutation of the unary encoding of $\mathrm{C}$ with equalities of rank $r$, in the form of a set of inequalities, and build an \SA\ refutation of the binary encoding of $\mathrm{C}$ of rank $r \log n$, by substituting terms $w_{x,a}$ in the former with  $\omega^{a_1}_{x,1}\cdots \omega^{a_{\log n}}_{x,\log n}$, where $a_1\ldots a_{\log n} = \mathrm{bin}(a)$, in the latter. Note that the equalities of the form
\[\sum_{a \in [n] : a_1 \ldots a_{\log n}=\mathrm{bin}(a)} \omega^{a_1}_{x,1}\cdots \omega^{a_{\log n}}_{x,\log n} = 1\] follow from the equalities (\ref{eq:LinSA2}). Further, inequalities of the form $\omega^{a_1}_{x,1}\ldots \omega^{a_{\log n}}_{x,\log n} \leq \nu_{x,a}$ follow since $\omega_{x,j}\overline{\omega}_{x,j}=0$ for each $j \in [\log n]$.
\end{proof}
The unary \emph{Ordering Principle ($\OP$) with equality} has the following set of \SA\ axioms:
\[
	\begin{array}{c}
	\mathit{self}: v_{i, i} = 0  \quad \forall \; i \in n\\
	\mathit{trans}: v_{i, k} - v_{i, j} - v_{j, k} + 1 \geq 0 \quad \forall \; i, j, k \in [n] \\
	\mathit{impl}: v_{i, j} - w_{i, j} \geq 0 \quad \forall \; i, j \in [n]  \\
	\mathit{lower}: \sum_{i \in [n]} w_{i, j} - 1 = 0 \quad \forall \; j \in [n]
	\end{array}
	\]
	Note that we need the $w$-variables since we use the equality form. Axioms of the form $\sum_{i \in [n]} x_{i, j} - 1 = 0$ made just from \textcolor{black}{$v$-}variables are plainly incompatible with, e.g., transitivity.
Strictly speaking Sherali-Adams is defined for inequalities only. An equality axiom $a = 0$ is simulated by the two inequalities $a \geq 0, -a \geq 0$, which we refer to as the 
	\emph{positive} and \emph{negative} instances of that axiom, respectively. 	
	 Also, note that we have used $v_{i,j}+\overline{v}_{i,j}=1$ to derive this formulation. We call 
		two product terms \emph{isomorphic} if one product term can be gotten from the other by relabelling the indices appearing in the subscripts \textcolor{black}{by a permutation}.
	
	\begin{theorem}
		The \SA\ rank of the $\OP$ with equality is at most $2$ and \SA\ size at most polynomial in $n$.
\label{theorem:8}
	\end{theorem}
\begin{proof}
	Note that if the polytope $\mathcal{P}^{\OP}_2$ is nonempty there must exist a point where any isomorphic variables are given the same value. We can find  such a point by averaging an asymmetric valuation over all permutations of $[n]$.\\
	So suppose towards a contradiction there is such a symmetric point. \textcolor{black}{First note $v_{i,i}=w_{i,i}=0$ by \textit{self} and \textit{impl}.}
	We start by lifting the $j$th instance of \textit{lower} by $v_{i, j}$ to get
	\[
		w_{i, j}v_{i, j} + \sum_{k \neq i, j}  w_{k, j} v_{i, j} = v_{i, j}.
	\]
	Equating (by symmetry {\color{black} with respect to $k$}) the product terms $w_{k, j} v_{i, j}$ this is actually
	\[
	w_{i, j}v_{i, j} + (n-2)  w_{k, j} v_{i, j} = v_{i, j}.
	\]
	Lift this by $w_{k, j}$ to get
	\[
	w_{k, j} w_{i, j}v_{i, j} + (n-2)  w_{k, j} v_{i, j} = w_{k, j} v_{i, j}.
	\]
	We can delete the leftmost product term by proving it must be $0$. Let us take an instance of \textit{lower} lifted by $w_{k, j}v_{i, j}$ for any $k \neq i,j$ along with an instance of monotonicity $w_{k, j}w_{m, j} v_{i, j} \geq 0$ for every $m \neq j,k$:
	\begin{gather}
	w_{k, j}v_{i, j} \left(1 - \sum_{m \neq j} w_{m, j}\right)  + \sum_{m \neq j,k,i} w_{k, j}w_{m, j} v_{i, j} \nonumber\\
	= - \sum_{m \neq k, j} w_{k, j}w_{m, j} v_{i, j} + \sum_{m \neq j,k,i} w_{k, j}w_{m, j} v_{i, j} \nonumber\\
	= - w_{k, j}w_{i, j} v_{i, j} \label{multwitnesses}. 
	\end{gather}
	The left hand side of this equation is greater than $0$ so we can deduce $w_{k, j}w_{i, j} v_{i, j}=0$.
	
	This results in 
	\[
	(n-2) w_{k, j} v_{i, j} = w_{k, j} v_{i, j} \quad \text{which is} \quad w_{k, j} v_{i, j} = 0.
	\]
	We lift \textit{impl} by $w_{i, j}$ to obtain $w_{i, j} \leq w_{i, j}v_{i, j}$. Monotonicity gives us the
	opposite inequality and we can proceed as if we had the equality $w_{k, j} v_{k, j} = w_{k, j}$ {\color{black} (as we are using equality as shorthand for inequality in both directions)} .\\
	So repeating the derivation of $w_{k, j} v_{i, j} = 0$ for every $i \neq k$ and then adding  $w_{k, j} v_{k, j} = w_{k, j}$ gets us $\sum_m w_{k, j} v_{m, j} = w_{k, j}$. Repeating this again for every $k$ and summing up gives
	\[
	 0 = \sum_{k,m} w_{k, j} v_{m, j} - \sum_k w_{k, j} = \sum_{k,m}  w_{k, j} v_{m, j} - 1
	\]
	with the last equality coming from the addition of the positive \textit{lower} instance $\sum_{k} w_{k, j} - 1 = 0$. Finally adding the lifted \textit{lower} instance $v_{m, j} - \sum_k w_{k, j} v_{m, j}\textcolor{black}{=0}$ for every $m$ gives
	\begin{equation} \label{ordersum}
	\sum_m v_{m, j} = 1.
	\end{equation}
	By lifting the \textit{trans} axiom $v_{i, k} - v_{i, j} - v_{j, k} + 1 \geq 0$ by $v_{j, k}$ we get
	\begin{equation} \label{zeros}
	v_{i, k}v_{j, k} - v_{i, j} v_{j, k}  \geq 0.
	\end{equation}
\textcolor{black}{
	Now, due to a manipulation similar to \Cref{multwitnesses} using \Cref{ordersum}
	\begin{gather}
	v_{k, j}v_{i, j} \left(1 - \sum_{m \neq j} v_{m, j}\right)  + \sum_{m \neq j,k,i} v_{k, j}v_{m, j} v_{i, j} \nonumber\\
	= - \sum_{m \neq k, j} v_{k, j}v_{m, j} v_{i, j} + \sum_{m \neq j,k,i} v_{k, j}v_{m, j} v_{i, j} \nonumber\\
	= - v_{k, j}v_{i, j} v_{i, j}\\
	= - v_{k, j}v_{i, j} \label{multwitnesses2}. 
	\end{gather}
	}

	 Thus,	$v_{i, k}v_{j, k}$ must be zero whenever $i \neq j$. Along with \Cref{zeros} we derive $v_{i, j} v_{j, k} = 0$.
	 {Noting $v_{i, j}v_{j, i}=0$ follows from \textit{trans} and \textit{self}, we lift}
	\Cref{ordersum} by $v_{j, x}$ for some $x$ \textcolor{black}{to} get
	\[
	 v_{j, x}\sum_m v_{m, j} = \sum_{m \neq x, j} v_{m, j}v_{j, x} = v_{j, x}
	\]
	where we know the left hand side is zero (\Cref{zeros}). Thus we can derive $v_{i, j} = 0$ for any $i$ and $j$, resulting in a contradiction when combined with \Cref{ordersum}.
	\end{proof}
	Before we derive our corollary, let us explicitly give the \SA\ axioms of $\bOP$.
	\[
	\begin{array}{c}
	\mathit{self}: \nu_{i, i} = 0  \quad \forall \; i \in n\\
	\mathit{trans}: \nu_{i, k} - \nu_{i, j} - \nu_{j, k} + 1 \geq 0 \quad \forall \; i, j, k \in [n] \\
	\mathit{impl}: \sum_{i\in [\log n]} \omega^{1-a_i}_{i,j} + \nu_{j,a} \geq 0 \quad \forall \; j \in [n] \\
		\mbox{where $a_1\ldots a_{\log n} = \mathrm{bin}(a)$} \\
	\end{array}
	\]
	\begin{corollary}
		The binary encoding of the Ordering Principle, $\bOP$, has \SA\ rank at most $2\log n$ and \SA\ size at most polynomial in $n$.
			\label{corollary:9}
	\end{corollary}
	\begin{proof}
		Immediate from Lemma \ref{lemma:binconv}.
	\end{proof}

\section{\SA+Squares}
\label{sec:SA+AS}
In this section we consider a proof system, \SA+Squares, based on inequalities of multilinear polynomials. We now consider axioms as degree-1 polynomials in some set of variables and refutations as polynomials in those same variables. Then this system is gotten from \SA\ by allowing addition of (linearised) squares of polynomials. In terms of strength this system will be strictly stronger than \SA\ and at most as strong as Lasserre (also known as Sum-of-Squares), although we do not at this point see \textcolor{black}{an exponential} separation between \SA+Squares and Lasserre.
See \cite{Lasserre2001,laurent01comparison,SoS-survey} for more on the Lasserre proof system \textcolor{black}{and \cite{Lauria2017} for tight degree lower bound results}.

Consider the polynomial $w_{i,j} v_{i,j} - w_{i,j} v_{i,k}$. The square of this is 
\[
w_{i,j} v_{i,j}w_{i,j} v_{i,j} + w_{i,j} v_{i,k} w_{i,j} v_{i,k} - 2 w_{i,j} v_{i,j} w_{i,j} v_{i,k}.
\]
Using idempotence this linearises to $w_{i,j} v_{i,j} + w_{i,j} v_{i,k} - 2 w_{i,j} v_{i,j} v_{i,k}$. Thus we know that this last polynomial is non-negative for all $0/1$ settings of the variables.\\
A \emph{degree-$d$} \SA+Squares refutation of a set of linear inequalities (over terms) $q_1 \geq 0, \ldots, q_x \geq 0$ is an equation of the form
\begin{equation} \label{eq:SA+Sref}
 \sum_{i = 1}^x p_i q_i + \sum_{i = 1}^y r_i^2 = -1
\end{equation}
where the $p_i$ are polynomials with nonnegative coefficients and the degree of the polynomials $p_i q_i, r_i^2$ is at most $d$. 
We want to underline that we now consider a (product) term like $w_{i,j} v_{i,j} v_{i,k}$ as a product of its constituent variables, that is genuinely a term in the sense of part of a polynomial. This is opposed to the preceding sections in which we viewed it as a single variable {\color{black} $Z_{w_{i,j} \wedge v_{i,j} \wedge v_{i,k}}$}.
The translation from the degree discussed here to \SA\ rank previously introduced may be paraphrased by ``$\mathrm{rank}=\mathrm{degree}-1$''.

We note that the unary $\PHP^{n+1}_n$ becomes easy in this stronger
proof system (see, e.g., Example 2.1 in \cite{russians}) while \textcolor{black}{we shall see that} the $\LOP$ remains hard {\color{black} (in terms of degree)}. The following is based on Example 2.1 in \cite{russians}.
\begin{theorem}
	The $\BinPHP^{n+1}_n$ has an $\SA+\mathrm{Squares}$ refutation of degree $2 \log n + 1$ and size $O(n^3)$.
	\label{thm:bin-php-sa-squares}
\end{theorem}
\begin{proof}
	For short let $m = n + 1$ denote the number of pigeons. We begin by squaring the polynomial
	\[
	1-\sum_{i=1}^{m} \prod_{j = 1}^{\log n} \omega_{i,j}^{a_j}
	\]
	to get the degree $2 \log n$, size quadratic in $m$ inequality
	\begin{equation}
	1 - 2 \sum_{i=1}^{m} \prod_{j = 1}^{\log n} \omega_{i,j}^{a_j} + \sum_{1\leq i, i' \leq m} \left( \prod_{j = 1}^{\log n} \omega_{i,j}^{a_j} \right) \left( \prod_{j = 1}^{\log n} \omega_{i',j}^{a_j} \right)  \geq0\label{eq:hole-squared}
	\end{equation}
	for every hole $a \in [n]$. On the other
	hand, by lifting each axiom 
	\[
	\sum_{j=1}^{\log n} \omega_{i,j}^{1 - a_j} + \sum_{j=1}^{\log n}\omega_{i',j}^{1 - a_j} \geq 1 \qquad \text{(whenever $i \neq i'$)}
	\]
	by $\left( \prod_{j = 1}^{\log n} \omega_{i,j}^{a_j} \right) \left( \prod_{j = 1}^{\log n} \omega_{i',j}^{a_j} \right)$ we find $0 \geq \left( \prod_{j = 1}^{\log n} \omega_{i,j}^{a_j} \right) \left( \prod_{j = 1}^{\log n} \omega_{i',j}^{a_j} \right)$, in degree $2 \log n + 1$. Adding these inequalities to (\ref{eq:hole-squared}) gives
	\begin{equation*}
	1 -  \sum_{i=1}^{m} \prod_{j = 1}^{\log n} \omega_{i,j}^{a_j} \geq 0
	\end{equation*}
	in size again quadratic in $m$. Iterating this for every hole $a \in [n]$ we find 
	\begin{equation}\label{eq:holeineq}
	n - \sum_{a = 1}^n \sum_{i=1}^{m} \prod_{j = 1}^{\log n} \omega_{i,j}^{a_j} \geq 0
	\end{equation}
	in cubic size.
	
	Note that for any pigeon $i \in [m]$, we can find in SA the linearly sized equality
	\begin{equation} \label{eq:pigeneqmac}
	\sum_{a = 1}^n \prod_{j = 1}^{\log n} \omega_{i, j}^{a_j} = 1.
	\end{equation}
	in size linear in $n$.
	
	This is done by induction on the number of bits involved (the range of $j$ in the summation). For the base case of just $j = 1$ we clearly have
	\[
	\omega_{i, 1} + (1 - \omega_{i, 1}) = 1.
	\]
	{\color{black}  Now suppose that for} $k < \log n$, we have $\sum_{a \in [2^k]} \prod_{j = 1}^k \omega_{i, j}^{a_j} = 1$. Multiplying both sides by $1 = \omega_{i, (k + 1)} + (1 - \omega_{i, (k + 1)})$ gets the inductive step. The final term is of size $O(2^{\log n}) = O(n)$.
	
	Summing \ref{eq:pigeneqmac} for every such hole $i$ we find
	\begin{equation}\label{eq:pigineq}
	\sum_{i=1}^{m} \sum_{a = 1}^n \prod_{j = 1}^{\log n} \omega_{i, j}^{a_j} \geq m.
	\end{equation}
	Adding \ref{eq:pigineq} to \ref{eq:holeineq}, we get the desired contradiction,
	$n - m \geq 0$.
\end{proof}

This last theorem, combined with the exponential SA size lower bound given in Theorem \ref{theorem:size-bound}, shows us that SA+Squares is exponentially separated from SA in terms of size.\\

We now turn our attention to $\LOP$, whose \SA\ axioms we reproduce to refresh the reader's memory.
\[
	\begin{array}{c}
	\mathit{self}: v_{i, i} = 0  \quad \forall \; i \in n\\
	\mathit{trans}: v_{i, k} - v_{i, j} - v_{j, k} + 1 \geq 0 \quad \forall \; i, j, k \in [n] \\
	\mathit{impl}: v_{i, j} - w_{i, j} \geq 0 \quad \forall \; i, j \in [n]  \\
	\mathit{total}: v_{i,j} + v_{j,i} -1 \geq 0 \quad   \forall \; i \neq j \in [n] \\
	\mathit{lower}: \sum_{i \in [n]} w_{i, j} - 1 \geq 0 \quad \forall \; j \in [n]
	\end{array}
	\]
We give our lower bound for the unary $\LOP$ by producing a linear function $\val$ (which we will call a \emph{valuation}) from terms into $\mathbb{R}$ such that
\begin{enumerate}
 \item for each axiom $p \geq 0$ and every term $X$ with $deg(Xp) \leq d$ we have $\val(X p) \geq 0$, and 
 \item we have $\val(r^2) \geq 0$ whenever $deg(r^2) \leq d$.
\item $\val(1)=1$.
\end{enumerate}
The existence of such a valuation clearly implies that a degree-$d$ \SA+Squares
refutation cannot exist, as it would result in a contradiction when applied to both sides of \cref{eq:SA+Sref}. 

To verify that $\val(r^2) \geq 0$ whenever $deg(r^2) \leq d$ we show that the so-called \emph{moment-matrix} $\mathcal{M}_{\val}$ is positive semidefinite. The degree-$d$ moment matrix is
defined to be the symmetric square matrix whose rows and columns are
indexed by terms of size at most $d/2$ and each entry is the valuation of the product of the two terms indexing that entry.
Given any polynomial $\sigma$ of degree at most $d/2$ let $c$ be its coefficient vector. Then if $\mathcal{M}_v$ is positive semidefinite:
\[
\val(\sigma^2) = \sum_{deg(T_1), deg(T_2) \leq d/2} {c}(T_1) {c}(T_2) v(T_1 T_2) = {c}^\top \mathcal{M}_v {c} \geq 0.
\]
(For more on this see e.g. \cite{Lasserre2001}, section 2.)\\
\begin{theorem}
\label{thm:SA2lbLOP}
	There is no $\SA+\mathrm{Squares}$ refutation of the (unary) $\LOP$ with degree at most $(n-3)/2$.
\label{theorem:11}
\end{theorem}
\begin{proof}
For each term $T$, let $\val\left(T\right)$ be the probability
that $T$ is consistent with a permutation on the $n$ elements taken
uniformly at random or, in other words, the number of permutations
consistent with $T$ divided by $n!$. Here we view $w_{x,y}$ as equal to $v_{x,y}$. This valuation trivially satisfies
the lifts of the \textit{self}, \textit{trans} and \textit{total} axioms as they are satisfied by each permutation (linear order). It satisfies the lifts of the \textit{impl} axioms by construction. We now claim that the lifts of the \textit{lower} axioms (those containing only $w$ variables) 
of degree up to $\frac{n-3}{2}$ are also satisfied
by $v\left(.\right)$. Indeed, let us consider the lifting by $T$
of the \textit{lower} axiom for $x$
\begin{equation}
\sum_{y=1}^{n}T\textcolor{black}{w}_{x,y}\geq T.\label{eq:there is smaller}
\end{equation}
Since $T$ mentions at most $n-3$ elements, there must be at least
two $y_{1}\neq y_{2}$ that are different from all of them and from
$x$. For any permutation that is consistent with $T$, the probability
that each of the $y_{1}$ and $y_{2}$ is smaller than $x$ is precisely
a half, and thus
\[
\val\left(T\textcolor{black}{w}_{x,y_{1}}\right)+\val\left(T\textcolor{black}{w}_{x,y_{2}}\right)=\val\left(T\right).
\]
Therefore the valuation of the LHS of (\ref{eq:there is smaller})
is always greater than or equal to the valuation of $T$.

Finally, we need to show that the valuation is consistent with the non-negativity of (the
linearisation of) any squared polynomial. It is easy to see that the moment matrix for $\val$ can be written as
\[
\frac{1}{n!}\sum_{\sigma}V_{\sigma}V_{\sigma}^{T}
\]
where the summation is over all permutations on $n$ elements and
for a permutation $\sigma$, $V_{\sigma}$ is its characteristic vector.
The characteristic vector of a permutation $\sigma$ is a Boolean
column vector indexed by terms and whose entries are $1$ or $0$
depending on whether the respective index term is consistent or not
with the permutation $\sigma$. Clearly the moment matrix is positive semidefinite being
a sum of (rank one) positive semidefinite matrices.
\end{proof}
The previous theorem is interesting because a {\color{black} degree} upper bound in Lasserre of order $\sqrt{n}\log n$ is known for $\LOP$ \cite{potechin2020sum}. It is proved for a slightly different formulation of $\LOP$ from ours, but it is {\color{black} readily seen to be equivalent to our formulation and we provide the translation in the appendix}. Thus, Theorem \ref{theorem:11}, together with \cite{potechin2020sum}, shows a quadratic rank separation between \SA+Squares and Lasserre.

\section{Contrasting unary and binary encodings}
\label{sec:unbin}

To work with a more general theory in which to contrast the complexity of refuting the binary and unary versions of combinatorial principles, following Riis \cite{SorenGap} we consider principles which are expressible as first order formulas with no finite model in $\Pi_2$-form,  i.e.  as $\forall \vec x \exists \vec w \varphi(\vec x,\vec w)$ where $\varphi(\vec x,\vec y)$ is a formula built on a family of relations $\vec R$. 
For example, we already met the Ordering Principle, \textcolor{black}{the version of which we will give here} states that a finite partial order has a maximal element. Its negation can be expressed in $\Pi_2$-form as:  
\[\forall x,y,z \exists w \ \neg R(x,x) \wedge (R(x,y) \wedge R(y,z) \rightarrow R(x,z)) \wedge R(x,w).\]
This can be translated into a unsatisfiable CNF using a {\em unary encoding} of the witness, as already discussed in Section~\ref{sec:SAeqLOP}.

As a second example we consider  the Pigeonhole Principle which  states  that a total mapping from $[m]$ to 
$[n]$ has necessarily a collision when $m$ and $n$ are integers with  $m>n$. Following Riis \cite{SorenGap}, for $m=n+1$, the negation of its  relational form can be expressed as a $\Pi_2$-formula as 
$$\forall x,y,z \exists w \ \neg R(x,0) \wedge (R(x,z) \wedge R(y,z) \rightarrow x=y) \wedge R(x,w)$$ 
and its usual unary and  binary propositional  encoding \textcolor{black}{have already been introduced}.
%
Notice that in the case of Pigeonhole Principle, the existential witness $w$ to the type \emph{pigeon} is of the distinct type \emph{hole}. Furthermore, pigeons only appear on the left-hand side of atoms $R(x,z)$ and holes only appear on the right-hand side. 
For the Ordering Principle instead, the transitivity axioms effectively enforce the type of $y$ appears on both the left- and right-hand side of atoms $R(x,z)$. This accounts for why, in the case of the Pigeonhole Principle, we did not need to introduce any new variables to give the binary encoding, yet for the Ordering Principle a new variable $w$ appears.

\subsection{Binary encodings of principles versus their \textcolor{black}{unary} functional encodings}
\label{subsec:func}
Recall the unary functional encoding of a combinatorial principle $\mathrm{C}$, denoted $\mathrm{Un}$-$\mathrm{Fun}$-$\mathrm{C}(n)$, replaces the big clauses from $\mathrm{Un}$-$\mathrm{C}(n)$, of the form $v_{i,1} \vee \ldots \vee v_{i,n}$, with $v_{i,1} + \ldots + v_{i,n} = 1$, where addition is made on the natural numbers. This is equivalent to augmenting the axioms $\neg v_{i,j} \vee \neg v_{i,k}$, for $j \neq k \in [n]$.
\begin{lemma}
\label{lem:bin-fun}
Suppose there is a Resolution refutation of $\mathrm{Bin}$-$\mathrm{C}(n)$ of size $S(n)$. Then there is a Resolution refutation of $\mathrm{Un}$-$\mathrm{Fun}$-$\mathrm{C}(n)$ of size at most $n^2\cdot S(n)$.
\end{lemma}
\begin{proof}
Take a decision DAG $\pi'$ for $\mathrm{Bin}$-$\mathrm{C}(n)$, where, without loss of generality, $n$ is even, and consider the point at which some variable $\nu_{i,j}$ is questioned. Each node in $\pi'$ will be expanded to a small tree in $\pi$, which will be a decision DAG for $\mathrm{Un}$-$\mathrm{Fun}$-$\mathrm{C}(n)$. The question ``$\nu_{i,j}?$'' in $\pi$ will become a sequence of questions $v_{i,1},\ldots,v_{i,n}$ where we stop the small tree when one of these is answered true, which must eventually happen. Suppose $v_{i,k}$ is true. If the $j$th bit of $k$ is $1$ we ask now all $v_{i,b_1},\ldots,v_{i,b_{\frac{n}{2}}}$, where $b_1,\ldots,b_{\frac{n}{2}}$ are precisely the numbers in $[n]$ whose $j$th bit is $0$. All of these must be false. Likewise, if the $j$th bit of $k$ is $0$ we ask all $v_{i,b_1},\ldots,v_{i,b_{\frac{n}{2}}}$, where $b_1,\ldots,b_{\frac{n}{2}}$ are precisely the numbers whose $j$th bit is $1$. All of these must be false. We now unify the branches on these two possibilities, forgetting any intermediate information. (To give an example, suppose $j=2$. Then the two outcomes are $\neg v_{i,1} \wedge \neg v_{i,3} \wedge  \ldots \wedge \neg v_{i,n-1}$ and $\neg v_{i,2} \wedge \neg v_{i,4} \wedge  \ldots \wedge \neg v_{i,n}$.) Thus, $\pi'$ gives rise to $\pi$ of size $n^2\cdot S(n)$ and the result follows.
\end{proof}

\subsection{The Ordering Principle in binary}

Recall the Ordering Principle whose binary formulation $\mathrm{Bin}$-$\mathrm{OP}_n$ we met in Section~\ref{sec:SAeqLOP}.
\begin{lemma}
$\mathrm{Bin}$-$\mathrm{OP}_n$ has refutations in Resolution of \textcolor{black}{size $O(n^3)$}.
\label{lem:bi-op-res}
\end{lemma}
\begin{proof}
We follow the well-known proof for the unary version of the Ordering Principle, from \cite{DBLP:journals/acta/Stalmarck96}. Consider the domain to be $[n]=\{1,\ldots,n\}$. At the $i$th stage of the decision DAG we will find a maximal element, ordered by $R$, among $[i]=\{1,\ldots,i\}$. That is, we will have a CNF record of the \emph{special} form 
\[\textcolor{black}{\neg \nu_{j,1} \wedge \ldots \wedge \neg  \nu_{j,j-1} \wedge \neg  \nu_{j,j+1} \wedge \ldots \wedge \neg  \nu_{j,i}}\]
for some $j \in [i]$. The base case $i=1$ is trivial. Let us explain the inductive step. From the displayed CNF record above we ask the question $ \nu_{j,i+1}?$ If $ \nu_{j,i+1}$ is true, then ask the sequence of questions $ \nu_{i+1,1},\ldots, \nu_{i+1,i}$, all of which must be false by transitivity \textcolor{black}{(the case $i=j$ uses irreflexivity too)}. Now, by forgetting information, we uncover a new CNF record of the special form. Suppose now $ \nu_{j,i+1}$ is false. Then we equally have a new CNF record again in the special form. Let us consider the size of our decision tree so far. There are $n^2$ nodes corresponding to special CNF records and navigating between special CNF records involves a path of length $n$, so we have a DAG of size $n^3$. Finally, at $i=n$, we have a CNF record of the form
\[\textcolor{black}{\neg \nu_{j,1} \wedge \ldots \wedge \neg \nu_{j,j-1} \wedge \neg  \nu_{j,j+1} \wedge \ldots \wedge \neg  \nu_{j,n}}.\]
Now we expand a tree questioning the sequence $\omega_{j,1},\ldots,\omega_{j,\log n}$, and discover each leaf labels a contradiction of the clauses of the final type. We have now added $n \cdot 2^{\log n}$ nodes, so our final DAG is of size at most $n^3+n^2$. 
\end{proof}




\section{Final remarks}
\label{sec:fin-rem}
\color{black}

\medskip

In this paper we started a systematic study of binary encodings of combinatorial principles in proof complexity.    Various questions arise directly from our exposition.  
Primarily, there is the question as to the optimality of our lower bounds for the binary encodings of $k$-Clique and the (weak) Pigeonhole Principle. In terms of the strongest refutation system $\RES(s)$ (largest $s$) for which we can prove superpolynomial bounds, then it is not hard to see that our method can go no further than $s= o((\log\log n)^\frac{1}{3})$ for the former, and \textcolor{black}{$s=O(\log^{\frac{1}{2}-\epsilon}n)$} for the latter. This is because we run out of space with the random restrictions as they become nested in the induction. We have no reason, however to think that our results are truly optimal, only that another method is needed to improve them.

A second question about binary encodings  concerns width and rank.  From our work it holds that in $\SA$ the unary  encoding can be harder than binary with respect to rank.  One might question whether  the same hold for Resolution width.  Are there formulas that require large width in the unary encoding, but can refuted in small width in the binary encoding? Notice that  in the other direction a large separation is not  possible. In particular it is straighforward  to see that  if the unary version of a formula $F$ over $n$ variables has  Resolution refutations of size $S$  and width $w$, then the binary version of $F$ has Resolution refutations of size $Sw^{\log n}$ and width $w \log n$.   

Other questions concern to what extent the converses of our lemmas might hold. 
The converse of Lemma \ref{lem:bin-fun} (even for $n^2$ replaced by some sublinear polynomial) is false. For example, consider the very weak Pigeonhole Principle of \cite{BussP97}. However, this example is somewhat disingenuous as the parameter $n$ is no longer polynomially related to the number of pigeons $m$ and the size of the clause set.

Finally an important question, not strictly regarding binary encodings, is the relative efficiency of SA+Squares with respect to Lasserre.
Is there a meaningful  size separation between $\SA+\mathrm{Squares}$ and Lasserre? Is Lasserre strictly stronger? At present we know only  the quadratic rank separation implied by our $\Omega(n)$ (Theorem \ref{thm:SA2lbLOP}) lower bound in  $\SA+\mathrm{Squares}$ and Potechin's upper $O(\sqrt{n})$ upper bound in Lasserre  for $\LOP$ .

\section*{Acknowledgments}
We are grateful to Ilario Bonacina for reading a preliminary version of this work and addressing  us some useful comments and observations. We are further grateful to several anonymous reviewers for  detailed corrections and comments. 

\color{black}

\section{Appendix}

\subsection{Potechin's encoding of $\LOP$}

Potechin provides a $O(\sqrt{n}\log n)$ upper bound in Lasserre for the following formulation of the linear ordering principle, which we purposefully give in the variables $x_{i,j}$ instead of our $v_{i,j}$.
\begin{gather*}
x_{i, j} + x_{j, i} = 1 \qquad \text{ for all distinct } i, j \in [n] \\
x_{i, j} x_{j, k} (1 - x_{i, k}) = 0 \qquad \text{ for all distinct } i, j, k \in [n] \\
\sum_{i \in [n], i \neq i} x_{i, j} = 1 + z_j^2
\end{gather*}
\noindent Note that anything we can prove using transitivity of the form $x_{i, j} x_{j, k} (1 - x_{i, k}) = 0$ we can prove using $v_{i, k} - v_{i, j} - v_{j, k}  \geq -1$.
That $v_{i, j} v_{j, k} \geq v_{i, j} v_{j, k} v_{i, k}$ comes from monotonicity, and the opposite inequality comes from lifting by $v_{i, j} v_{j, k}$:
\[
- v_{i, j} v_{j, k} \leq v_{i, j} v_{j, k} v_{i, k} - 2 v_{i, j} v_{j, k} \implies v_{i, j} v_{j, k} \leq v_{i, j} v_{j, k} v_{i, k}. 
\]
Potechin's proof moves along the following lines. Define an operator $E$ on terms that behaves the same
as the $\val$ used in \Cref{theorem:11}, but
\begin{enumerate}
	\item If some $z_j$ appears with degree $1$ in $T$, then $E[T] = 0$, and
	\item If $T$ is of the form $z_j^2 T'$ for some $j$ and $T'$, $E[T] = E\left[ \left( \sum_{i \in [n], i \neq i} x_{i j} - 1\right) T'\right]$
\end{enumerate}
Potechin proves the following.
\begin{lemma}[Lemma 4.2 in \cite{potechin2020sum}]
 \label{lem:badpoly}
	There exists a polynomial $g$, only in the variables $x_{i, j}$ and of degree $O(\sqrt{n} \log n)$ such that
	\[
	E\left[\left( \sum_{i \neq j} x_{i, j} - 1 \right) g^2 \right] = \val \left(\left( \sum_{i \neq j} x_{i, j} - 1 \right) g^2 \right)  < 0.
	\]
\end{lemma}
\noindent Potechin then proves the following Lasserre identity using only the totality and transitivity axioms (which exist also in our formulation). Note $S_k$ is the symmetric group on the elements of $[k]$.
\begin{lemma}[Lemma 4.7 in \cite{potechin2020sum}]
	For all $A = \{i_1, i_2, \ldots, i_{k} \}\subseteq [n]$, there exists a degree $k + 2$ proof that
	\[
	\sum_{\pi \in S_{k}} \prod_{j = 1}^{k - 1} x_{i_{\pi(j)} i_{\pi(j + 1)}} = 1.
	\]
\end{lemma}
Finally, Potechin proves that the `symmetric group average' of a polynomial can be shown to be equal to its valuation.
\begin{lemma}[Lemma 4.8 in \cite{potechin2020sum}]
	For any polynomial $p$ of degree $d$ in the variables $x_{i j}$, there exists a proof of at most degree $3d + 2$ that
	\[
	\frac{1}{n!} \sum_{\pi \in S_n} \pi(p) = \val(p)
	\]
	(where the action of $S_n$ is to permute the indices in the monomials of $p$).
\label{lem:potechin-4.8}
\end{lemma}
Lemma \ref{lem:badpoly} and~\ref{lem:potechin-4.8} together furnish a Lasserre refutation of the required form.

\color{black} 

\subsection{Recapitulation of the unary and binary encodings of the main principles}

\begin{figure}
\color{black}
\begin{tabular}{|c|c|c|}
\hline
principle & unary case & binary case \\
\hline
 & 
$\neg v_{i,a} \vee \neg v_{j,b}$ & $
(\omega^{1-a_1}_{i,1}\vee \ldots \vee \omega^{1-a_{\log n}}_{i,\log n})$
\\
& whenever $\neg E((i,a),(j,b))$ & $\vee$ \\
$\mathrm{(Bin\mbox{-})Clique^k_n}$ & and & $(\omega^{1-b_1}_{j,1}\vee \ldots \vee \omega^{1-b_{\log n}}_{j,{\log n}})$ \\
& $\bigvee_{a \in [n]} v_{i,a}$ & whenever $\neg E((i,a),(j,b))$ \\
& for each block $i \in [k]$  & where binary representations are \\
& & $a = a_1 \ldots  a_{\log n}$ \\
& & $b = b_1 \ldots  b_{\log n}$ \\
\hline
& $\neg v_{i,a} \vee \neg v_{j,a}$ & $
(\omega^{1-a_1}_{i,1}\vee \ldots \vee \omega^{1-a_{\log n}}_{i,\log n})$
\\
& whenever $i \neq j$ & $\vee$ \\
$\mathrm{(Bin\mbox{-})PHP}^m_n$  & and & $(\omega^{1-a_1}_{j,1}\vee \ldots \vee \omega^{1-a_{\log n}}_{j,{\log n}})$ \\
& $\bigvee_{a \in [n]} v_{i,a}$ & whenever $i \neq j$ \\
& for each pigeon $i\in [m]$  & where binary representation is \\
& & $a = a_1 \ldots  a_{\log n}$ \\
\hline
& $\neg v_{i,i}$ &
$\neg \nu_{i,i}$ for all $i \in [n]$
\\
&  for all $i \in [n]$ & $\neg \nu_{i,j} \vee \neg \nu_{j,k} \vee \nu_{i,k}$ \\
$\mathrm{(Bin\mbox{-})OP}n$  & $\neg v_{i,j} \vee \neg v_{j,k} \vee v_{i,k}$  & for all $i,j,k \in [n]$ \\
& for all $i,j,k \in [n]$ & $\bigvee_{i \in [n]} \nu_{i,j}$ for all $j\in [n]$\\
& and & and \\
& $\bigvee_{a \in [n]} v_{i,a}$  & $(\omega^{1-a_1}_{i,1}\vee \ldots \vee \omega^{1-a_{\log n}}_{i,\log n} \vee \nu_{a,i})$ \\
& for all $a \in [n]$ & for all $a \in [n]$ whose binary representation is \\
&  & $a_1 \ldots  a_{\log n}$ \\
\hline
\end{tabular}
\caption{Recapitulation of the unary and binary encodings of the main principles.}
\label{fig:last-recap}
\end{figure}
\end{document}